\documentclass[aps,pre,twocolumn,superscriptaddress]{revtex4-1}

\usepackage{euscript,amsmath,amssymb,amsfonts,graphicx,bm,amsthm,mathtools}

  \usepackage{paralist}
  \usepackage{epstopdf}
  \usepackage{graphics} 
 \usepackage[latin1]{inputenc}
\usepackage[caption=false]{subfig}
\usepackage{soul}

\usepackage{latexsym,amssymb,enumitem,amsmath,amsthm} 
  \usepackage{graphicx} 
  \usepackage{tikz}
  \usepackage[hidelinks]{hyperref}
\usepackage{latexsym,amssymb,amsmath} 
\usepackage{pgfplots}
\usepackage{soul,xcolor}

\newcommand{\dd}{\textup{d}}
\def\eps{\varepsilon}
\def\E{\mathbb{E}}
\def\P{\mathbb{P}}
\def\R{\mathbb{R}}
\newcommand{\Markov}[2]{\underset{#1}{\overset{#2}{\rightleftharpoons}}}
\def\tmin{t_{\min}}
\def\O{\mathcal{O}}
\def\taut{\widetilde{\tau}}

\def\tc{t_{\textup{c}}}
\DeclarePairedDelimiter\floor{\lfloor}{\rfloor}

\newtheorem{assumption}{Assumption}
\newtheorem{theorem}{Theorem}
\newtheorem{lemma}{Lemma}
\newtheorem{proposition}{Proposition}
\newtheorem{corollary}{Corollary}
\theoremstyle{plain}
\theoremstyle{remark}

\theoremstyle{definition}

\begin{document}


\title[]{Universality and ambiguity in extremes of anomalous diffusion}


\author{Sean D. Lawley}
\email[]{lawley@math.utah.edu}
\affiliation{University of Utah, Department of Mathematics, Salt Lake City, UT 84112 USA}


\date{\today}

\begin{abstract}
Many biophysical processes begin when the fastest searcher finds a target out of many random searchers, which is called an extreme or fastest first passage time (fFPT). In some models, (i) the fFPT vanishes logarithmically as the number of searchers grows, and (ii) the fFPT can be faster for subdiffusive search compared to normal diffusion. Though mathematically rigorous, the relevance of (i) and (ii) to actual physical systems is suspect since their derivations involve searchers which move with unbounded speed. Indeed, we previously proved that the fFPT for searchers with bounded speed converges exponentially to a strictly positive minimal search time as the number of searchers grows. In this paper, we study fFPTs for a broad class of anomalous and normal diffusion models with bounded or unbounded speed. These models include scaled Brownian motion, Riemann-Liouville fractional Brownian motion, and fractional Brownian motion. For all of these models, we show that the fFPT decays logarithmically in the number of searchers and that subdiffusion can be faster than normal diffusion (we further show that superdiffusion can be slower than normal diffusion). In this sense, features (i) and (ii) are rather universal. On the other hand, we show that the parameter regimes in which (i) and (ii) are valid depend on the particulars of the individual model, and thus ambiguities remain in the relevance of these features to specific physical systems.
\end{abstract}


\maketitle

\section{\label{sec:intro}Introduction}

Timescales are often studied in terms of first passage times (FPTs), which generically describe the time it takes a ``searcher'' to find a ``target'' \cite{grebenkov2024target}. Mathematically, if $X(t)$ denotes the one-dimensional position of a searcher at time $t\ge0$ that starts at $X(0)=0$, then the FPT to some level $L>0$ is
\begin{align}\label{eq:tau}
    \tau
    =\inf\{t\ge0:X(t)\ge L\}.
\end{align}
A variety of biophysical processes are triggered by the FPT of the fastest searcher out of many searchers \cite{lawley2024competition}. Such ``fastest FPTs'' (fFPTs) are typically modeled as the minimum of $N\gg1$ FPTs,
\begin{align}\label{eq:TN}
    T_N
    :=\min\{\tau_1,\dots,\tau_N\},
\end{align}
where $\{\tau_1,\dots,\tau_N\}$ are $N$ independent realizations of the FPT $\tau$ in \eqref{eq:tau}. The fFPT in \eqref{eq:TN} is also called an ``extreme FPT'' since it is an extreme statistic \cite{coles2001}, though this term also encompasses slowest FPTs \cite{lawley2023slow}.

If the searchers move by Brownian motion with diffusivity $D>0$, then the mean fFPT decays as \cite{weiss1983}
\begin{align}\label{eq:normal}
    \E[T_N]
    \sim\frac{L^2}{4D\ln N}\quad\text{as }N\to\infty,
\end{align}
where $f\sim g$ denotes $\lim f/g=1$. 
The result \eqref{eq:normal} holds in very general circumstances for diffusive motion if $L$ is replaced by an appropriate geodesic distance \cite{lawley2020uni}. 

While diffusion is the prototypical model of stochastic search, anomalous diffusion has been observed in many natural and engineered systems \cite{metzler2000}, and many works have sought to understand FPTs of anomalous diffusion \cite{rangarajan2000anomalous, condamin2007, condamin2008, kosztolowicz2022first, gomez2024first}. Anomalous diffusion is marked by the following power law growth of the mean-squared displacement,
\begin{align}\label{eq:msd}
    \E[(X(t))^2]
    =2D\tc(t/\tc)^\alpha,\quad \alpha>0,
\end{align}
for some characteristic timescale $\tc>0$. 
Anomalous diffusion is classified by the exponent $\alpha$ as subdiffusion if $\alpha\in(0,1)$, superdiffusion if $\alpha\in(1,2)$, ballistic if $\alpha=2$, and hyperballistic if $\alpha>2$, whereas the case $\alpha=1$ is normal diffusion (i.e.\ not anomalous).

For a given anomalous diffusion exponent $\alpha$, a variety of models of stochastic motion yield the nonlinear mean-squared displacement in \eqref{eq:msd}. A common model of subdiffusion is a (time-)fractional Fokker-Planck equation \cite{metzler1999}, and it was proven that the corresponding mean fFPT decays as \cite{lawley2020sub}
\begin{align}\label{eq:sub}
    \E[T_N^{\mathrm{fFPe}}]
    \sim\frac{t_\alpha}{(\ln N)^{2/\alpha-1}}\quad\text{as }N\to\infty,
\end{align}
where $t_\alpha$ is a certain subdiffusive timescale (the superscript ``fFPe'' emphasizes that \eqref{eq:sub} holds for the fractional Fokker-Planck equation model of subdiffusion).

Consider the following two arguably non-physical features of \eqref{eq:normal} and \eqref{eq:sub}. 
\begin{enumerate}[label=(\roman*), start=1]
    \item The asymptotic in \eqref{eq:normal} implies that diffusive searchers can reach the target arbitrarily fast for sufficiently large $N$.
    \item The asymptotics in \eqref{eq:normal} and \eqref{eq:sub} imply that the fFPT can be faster for subdiffusion compared to normal diffusion.
\end{enumerate}
Though the technical mathematical validity of features (i) and (ii) is unquestioned (for Brownian motion and the fractional Fokker-Planck equation), their relevance to actual physical systems is suspect.

As detailed in Ref.~\cite{lawley2021pdmp}, feature (i) must break down for large enough $N$ if the searcher speed is bounded. Indeed, if the maximum searcher speed is $v<\infty$, then
\begin{align}\label{eq:lowerbound0}
    T_N\ge L/v>0\quad\text{for all }N\ge1.
\end{align}
In fact, if searchers move with fixed speed $v$ and change directions at exponentially distributed times, then $T_N$ converges exponentially to $L/v>0$ as $N\to\infty$ with a Weibull distributed correction for finite $N$ \cite{lawley2021pdmp}. 
The discrepancy between the vanishing time in \eqref{eq:normal} and the lower bound in \eqref{eq:lowerbound0} stems from the infinite speed of propagation of solutions to the diffusion equation \cite{keller2004, kuske1997, joseph1989}. To elaborate, let $p(x,t)$ be the probability density for the searcher position $X(t)=x$ given that $X(0)=0$. If searchers move by Brownian motion, then their position density is Gaussian,
\begin{align}\label{eq:pg}
    p(x,t)
    =\frac{1}{\sqrt{4\pi Dt}}\exp\bigg(\frac{-x^2}{4Dt}\bigg)>0,\quad x\in\R,\,t>0,
\end{align}
and thus $p(x,t)>0$ for all $x\in\R$ if $t>0$. Hence, the speed of Brownian search is unbounded in the sense that a searcher has a strictly positive probability of being arbitrarily far from its starting position at any positive time. In contrast, if the maximum searcher speed is $v<\infty$, then 
\begin{align}\label{eq:p0}
    p(x,t)=0\quad\text{if }|x|>vt.
\end{align}
See Ref.~\cite{lawley2021pdmp} for details on the discrepancy between \eqref{eq:pg} and \eqref{eq:p0}. This discrepancy was also recently studied in the context of fFPTs and experiments involving photon transport through a scattering medium \cite{carroll2025measurements}. An interesting and detailed analysis of this discrepancy also recently appeared in Ref.~\cite{grebenkov2026fastest}.

Feature (ii) is counterintuitive because subdiffusion is typically understood to be slower than normal diffusion. Is feature (ii) an artifact of the infinite speed of propagation (akin to the discrepancy between \eqref{eq:pg} and \eqref{eq:p0})? Perhaps feature (ii) is an idiosyncrasy peculiar to the fractional Fokker-Planck equation model of subdiffusion? Importantly, the potential for subdiffusive search to be faster than diffusive search is not a mere theoretical curiosity. Indeed, some have suggested that biological cells benefit from their crowded internal state because crowding induces subdiffusion which enhances the intracellular search processes needed for protein complex formation and signal propagation \cite{guigas2008}.

In this paper, we investigate the fFPTs of some anomalous diffusion processes. We start in section~\ref{sec:gaussian} by considering a broad class of Gaussian processes with the power law mean-squared displacement in \eqref{eq:msd}. This class includes scaled Brownian motion (sBm), Riemann-Liouville fractional Brownian motion (RLfBm), and fractional Brownian motion (fBm), which are canonical models of anomalous diffusion \cite{lim2002self} (both subdiffusion and superdiffusion). For this broad class of models, we prove that the $m$th moment of the fFPT decays logarithmically,
\begin{align}\label{eq:main1}
    \E[(T_N)^m]
    \sim
    (\tc)^m\bigg(\frac{L^2}{4D\tc\ln N}\bigg)^{m/\alpha}\quad\text{as }N\to\infty.
\end{align}
Hence, the fFPT increases as $\alpha$ increases for large $N$. More precisely, the asymptotic expression in \eqref{eq:main1} is an increasing function of $\alpha$ if
\begin{align*}
    N
    >\exp\bigg(\frac{L^2}{4D\tc}\bigg).
\end{align*}
In this sense, subdiffusion can be faster than normal diffusion, which in turn can be faster than superdiffusion. 
However, these Gaussian models have unbounded speed in the sense that the probability density for the searcher position $X(t)=x$ is strictly positive at all $x\in\R$ at every strictly positive time $t>0$, 
\begin{align}\label{eq:pg2}
    p(x,t)
    =\frac{1}{\sqrt{4\pi D\tc(t/\tc)^\alpha}}\exp\bigg(\frac{-x^2}{4D\tc(t/\tc)^\alpha}\bigg)>0.
\end{align}

\begin{figure}
\centering
\includegraphics[width=1\linewidth]{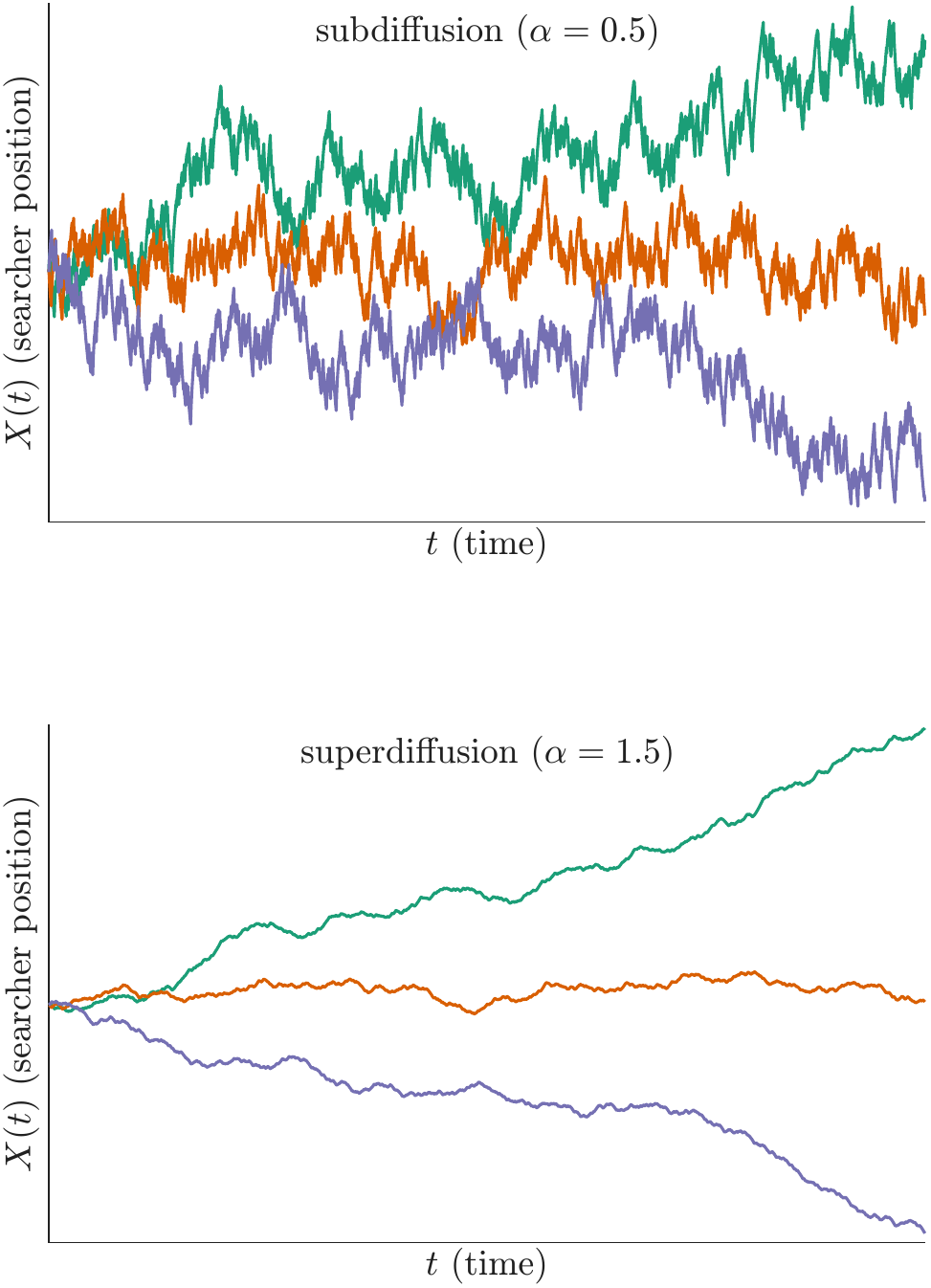}
\caption{Example trajectories of a bounded speed anomalous diffusion \cite{dean2021position}. The process is defined in \eqref{eq:finitespeedprocess}.}
\label{fig:trajectory}
\end{figure}

Therefore, in section~\ref{sec:finite}, we consider bounded speed analogs to sBm and RLfBm. The bounded speed analog to sBm is a time-change of the bounded speed analog to Brownian motion studied in Ref.~\cite{lawley2021pdmp}. The bounded speed analog to RLfBm was introduced in Ref.~\cite{dean2021position} and its fFPTs were recently studied numerically in Ref.~\cite{grebenkov2026fastest}. Figure~\ref{fig:trajectory} plots trajectories of the bounded speed RLfBm for $\alpha=1/2$ (top panel) and $\alpha=3/2$ (bottom panel). We prove that the fFPT for these bounded speed processes converges exponentially to a strictly positive search time $t_{\min}>0$ as $N\to\infty$, and we rigorously derive explicit formulas for all the moments of the fFPT and the fFPT probability distribution for large $N$. In particular, we prove that
\begin{align}\label{eq:main2}
    T_N
    \approx t_{\min}(1+\chi_N),
\end{align}
where $\chi_N$ is a certain Weibull type random variable (the explicit formula for $\chi_N$ and the precise sense of the approximation in \eqref{eq:main2} are given in section~\ref{sec:finite}).  

Furthermore, these bounded speed analogs depend on a dimensionless parameter $\eps>0$ and converge to their Gaussian counterparts (sBm and RLfBm) as $\eps\to0$. Hence, one expects that if $0<\eps\ll1$, then the fFPTs for such bounded speed processes are (a) well-approximated by \eqref{eq:main1} if $1\ll N\ll N_1$, and are (b) well-approximated by \eqref{eq:main2} if $N\gg N_2$, where $N_1\le N_2$ are some $\eps$-dependent thresholds which determine the transition between these two parameter regimes. We derive formulas for $N_1$ and $N_2$ in section~\ref{sec:finite}. We compare our theoretical predictions to numerical simulations in section~\ref{sec:numerics}. We conclude that features (i) and (ii) are not artifacts of unbounded search speed and that feature (ii) is not peculiar to the fractional Fokker-Planck model of subdiffusion. Nevertheless, our results show that the parameter regimes in which features (i) and (ii) are relevant depend on the specifics of the individual model.

\section{\label{sec:gaussian}Unbounded speed anomalous diffusions}

We now derive the fFPT asymptotics for a broad class of anomalous diffusion processes.

\subsection{\label{sec:assumptions}Assumptions}

Assumption~\ref{assumption} collects the assumptions that we make on search processes with unbounded speed.

\begin{assumption}\label{assumption}
    Assume $X=\{X(t)\}_{t\ge0}$ is a mean-zero Gaussian process with $X(0)=0$, $\E[(X(1))^2]>0$, and continuous sample paths which is self-similar with parameter $\alpha/2>0$. 
\end{assumption}

The assumption $\E[(X(1))^2]>0$ excludes the trivial case that $X(t)=0$ for all $t\ge0$ with probability one.

The assumption of self-similarity with parameter $\alpha/2>0$ means that $\{X(ct)\}_{t\ge0}$ and $\{c^{\alpha/2}X(t)\}_{t\ge0}$ have the same probability law for all $c>0$. Therefore, the mean-squared displacement of a process $X$ satisfying Assumption~\ref{assumption} is
\begin{align}\label{eq:msd1}
    \E[(X(t))^2]
    =2D\tc(t/\tc)^\alpha,\quad t\ge0,
\end{align}
where $\tc>0$ is any fixed time and
\begin{align*}
    D=\frac{\E[(X(\tc))^2]}{2\tc}>0.
\end{align*}
Depending on the parameter $\alpha>0$, the process $X$ can thus be subdiffusive ($\alpha\in(0,1)$), normal diffusive ($\alpha=1$), superdiffusive ($\alpha\in(1,2)$), ballistic ($\alpha=2$), or hyperballistic ($\alpha>2$).

\subsection{\label{sec:examples}Example processes}

There are many processes $X$ which satisfy Assumption~\ref{assumption}. We now describe three examples \cite{lim2002self}. In the following, let $W=\{W(t)\}_{t\in\R}$ denote a standard Brownian motion (i.e.\ satisfying $\E[(W(t))^2]=t$ for $t\ge0$). 

One example satisfying Assumption~\ref{assumption} is the following time change of $W$,
\begin{align}\label{eq:sbm}
    X_{\text{sBm}}(t)
    =\sqrt{2D}\,W(\tc(t/\tc)^\alpha),
\end{align}
which is called a scaled Brownian motion (sBm) \cite{lim2002self}. Note that $X_{\text{sBm}}$ is well-defined for all $\alpha>0$. Note also that, owing to self-similarity of Brownian motion, we have the following equality in distribution,
\begin{align*}
    X_{\text{sBm}}(t)
    =\sqrt{2D}\,W(\tc(t/\tc)^\alpha)
    =_\text{dist}\sqrt{2(t/\tc)^{\alpha-1}D}\,W(t),
\end{align*}
and therefore $X_{\text{sBm}}(t)$ can be understood as a Brownian motion with the following time-dependent diffusivity,
\begin{align}\label{eq:Doft}
    D(t)=(t/\tc)^{\alpha-1}D.
\end{align}

Another example satisfying Assumption~\ref{assumption}, which is a canonical model of anomalous diffusion, is a Riemann-Liouville fractional Brownian motion (RLfBm), which has the following representation,
\begin{align}\label{eq:RLfBm}
    X_{\text{RLfBm}}(t)
    &=\sqrt{\alpha}\sqrt{2 D\tc^{1-\alpha}}\int_0^t(t-s)^{\alpha/2-1/2}\,\dd W(s),
\end{align}
where the multiplicative factor $\sqrt{\alpha}$ ensures that $X_{\text{RLfBm}}$ has the mean-squared displacement in \eqref{eq:msd1}. 
Note that $X_{\text{RLfBm}}$ is well-defined for all $\alpha>0$.

Another canonical example satisfying Assumption~\ref{assumption} is a fractional Brownian motion (fBm), which has the following representation \cite{mandelbrot1968fractional}, 
\begin{align}\label{eq:fBm}
\begin{split}
    &X_{\text{fBm}}(t)
    =c_\alpha\sqrt{2 D\tc^{1-\alpha}}\bigg(\int_0^t(t-s)^{\alpha/2-1/2}\,\dd W(s)\\
    &+\int_{-\infty}^0\Big[(t-s)^{\alpha/2-1/2}-(-s)^{\alpha/2-1/2}\Big]\,\dd W(s)\bigg),
\end{split}    
\end{align}
where the multiplicative factor $c_\alpha$ ensures that $X_{\text{fBm}}$ has the mean-squared displacement in \eqref{eq:msd1} and is given by
\begin{align*}
    c_\alpha
    =\frac{\sqrt{\sin(\pi\alpha/2)\Gamma(\alpha+1)}}{\Gamma((\alpha+1)/2)},
\end{align*}
where $\Gamma(\cdot)$ is the gamma function. 
Unlike RLfBm, the parameter $\alpha$ is restricted to the interval $\alpha\in(0,2)$ for fBm. Note that RLfBm and fBm are commonly parameterized by the so-called Hurst exponent $H=\alpha/2$.

The probability density function $p(x,t)$ of sBm, RLfBm, and fBm is the same Gaussian function in \eqref{eq:pg2}. Hence, sBm, RLfBm, and fBm have the same Fokker-Planck equation with the time-dependent diffusivity in \eqref{eq:Doft},
\begin{align*}
    \frac{\partial}{\partial t}p(x,t)
    =(t/\tc)^{\alpha-1}D\frac{\partial^2}{\partial x^2}p(x,t),\quad x\in\R,\,t>0,
\end{align*}
and Dirac initial condition $p(x,0)=\delta(x)$. Nevertheless, we emphasize that sBm, RLfBm, and fBm are three distinct stochastic processes, unless $\alpha=1$, in which case sBm, RLfBm, and fBm each reduce to Brownian motion with diffusivity $D>0$. For instance, if $\alpha\neq1$, then only sBm is Markovian and only fBm has stationary increments.

\subsection{Short-time FPT asymptotics}

Suppose $X=\{X(t)\}_{t\ge0}$ satisfies Assumption~\ref{assumption}. Since $X(t)$ is Gaussian, the probability that $X(t)$ exceeds some $L>0$ satisfies
\begin{align}\label{eq:gim}
    \P(X(t)\ge L)
    \sim\sqrt{\frac{\sigma^2t^{\alpha}}{2\pi L^2}}\exp\Big(-\frac{L^2}{2\sigma^2t^{\alpha}}\Big)\quad\text{as }\frac{t^{\alpha}}{L^2}\to0^+,
\end{align}
where $\sigma^2=2D\tc^{1-\alpha}$, which is a direct consequence of \eqref{eq:pg2}. Let $\tau$ be the FPT to $L>0$,
\begin{align}\label{eq:tau1}
    \tau
    =\inf\{t\ge0:X(t)\ge L\}.
\end{align}
The following proposition states that $\tau$ has the same asymptotics as \eqref{eq:gim} on a logarithmic scale.
\begin{proposition}\label{prop:key}
If $X=\{X(t)\}_{t\ge0}$ satisfies Assumption~\ref{assumption} and $\tau$ is the FPT in \eqref{eq:tau1}, then
\begin{align}\label{eq:key}
\begin{split}
    \lim_{t\to0^+}t^{\alpha}\ln\P(\tau\le t)
    &=-\frac{L^2}{2\sigma^2}
    =-\tc^\alpha\frac{L^2}{4D\tc}<0.
\end{split}    
\end{align}    
\end{proposition}

Proposition~\ref{prop:key} is an intuitive consequence of \eqref{eq:gim} (the proof is in Appendix~\ref{sec:proofs}). Indeed, the fact that \eqref{eq:key} holds with $\P(\tau\le t)$ replaced by $\P(X(t)\ge L)$ follows from \eqref{eq:gim}. Furthermore, the inequality $\P(\tau\le t)\ge\P(X(t)\ge L)$ is immediate, and the probability that $X$ first hits $L$ and then backtracks substantially below $L$ before time $t$ is unlikely as $t^\alpha/L^2$ vanishes.

We emphasize that \eqref{eq:key} is a general result that holds for all processes $X$ satisfying Assumption~\ref{assumption}. In particular, \eqref{eq:key} holds for sBm, RLfBm, and fBm defined in section~\ref{sec:examples}. We show in the next subsection that the logarithmic asymptotic in \eqref{eq:key} yields all the moments of the fFPT out of $N\gg1$ independent and identically distributed (iid) searchers.

\subsection{Leading order fFPTs}

We now give a general theorem which yields all the fFPT moments of any process satisfying Assumption~\ref{assumption}. The proof is in Appendix~\ref{sec:proofs}. Throughout this paper, $f\sim g$ means $\lim f/g=1$.

\begin{theorem}\label{thm:uni}
Assume that $\{\tau_{n}\}_{n\ge1}$ is a sequence of iid realizations of a random variable $\tau$ satisfying
\begin{align}\label{eq:conditionb}
\lim_{t\to0^+}t^{\alpha}\ln\P(\tau\le t)=-C<0,
\end{align}
where $\alpha>0$ and $C>0$. Define
\begin{align}\label{eq:TN1}
    T_N:=\min\{\tau_1,\dots,\tau_N\},
\end{align}
and assume that
\begin{align}\label{eq:conditiona}
\E[T_{N}]
<\infty\quad\text{for some $N\ge1$}.
\end{align}
Then for any $m\ge1$,
\begin{align}\label{eq:res}
\E[(T_{N})^{m}]
\sim\Big(\frac{C}{\ln N}\Big)^{m/\alpha}
\quad\text{as }N\to\infty.
\end{align}
\end{theorem}

If $\tau$ is the FPT in \eqref{eq:tau1} and $X=\{X(t)\}_{t\ge0}$ satisfies Assumption~\ref{assumption}, then \eqref{eq:key} implies that \eqref{eq:conditionb} holds with
\begin{align*}
    C
    =\tc^\alpha\frac{L^2}{4D\tc}>0.
\end{align*}
Furthermore, if we assume that $\P(\tau>t)$ decays no slower than polynomially as $t\to\infty$, then \eqref{eq:conditiona} holds. For self-similar processes $X$, one generally expects that \cite{aurzada2022asymptotics}
\begin{align}\label{eq:persistence}
    \P(\tau>t)
    \propto t^{-\theta+o(1)}\quad\text{for $\theta>0$ as }t\to\infty,
\end{align}
where $\theta>0$ is termed the persistence exponent. If \eqref{eq:persistence} holds with $\theta>0$, then \eqref{eq:conditiona} holds for all $N\ge 1/\theta+1$. We have thus obtained the following corollary of Proposition~\ref{prop:key} and Theorem~\ref{thm:uni}.

\begin{corollary}\label{cor:decay}
    If $X=\{X(t)\}_{t\ge0}$ satisfies Assumption~\ref{assumption} and its FPT $\tau$ in \eqref{eq:tau1} satisfies \eqref{eq:persistence}, then for any $m\ge1$, the $m$th moment of $T_N$ in \eqref{eq:TN1} decays as
\begin{align}\label{eq:uniimplication}
        \E[(T_N)^m]
    \sim
    (\tc)^m\bigg(\frac{L^2}{4D\tc\ln N}\bigg)^{m/\alpha}\quad\text{as }N\to\infty.
\end{align}
\end{corollary}

It is straightforward to check that \eqref{eq:persistence} holds for the sBm in \eqref{eq:sbm} with $\theta=\alpha/2$. In addition, it is well-known that \eqref{eq:persistence} holds for fBm with $\theta=1-\alpha/2>0$ for $\alpha\in(0,2)$. Similarly, \eqref{eq:persistence} holds for RLfBm for some $\alpha$-dependent $\theta=\theta(\alpha)>0$ if $\alpha\in(0,2)$, though an explicit formula for $\theta(\alpha)$ is unknown for RLfBm \cite{aurzada2022asymptotics}. 

Therefore, if $\alpha\in(0,2)$ and $X$ is the sBm in \eqref{eq:sbm}, the RLfBm in \eqref{eq:RLfBm}, or the fBm in \eqref{eq:fBm} (and thus has mean-squared displacement in \eqref{eq:msd1}), then the fFPT moments decay according to \eqref{eq:uniimplication}. 
Hence, fFPTs decay faster for subdiffusion ($\alpha\in(0,1)$) than for normal diffusion ($\alpha=1$) or superdiffusion ($\alpha\in(1,2)$) for sufficiently large $N$.

\subsection{More detailed fFPT asymptotics}

Roughly speaking, \eqref{eq:key} means that
\begin{align*}
    \P(\tau\le t)
    \approx At^pe^{-L^2/(2\sigma^2 t^\alpha)}\quad\text{if }2\sigma^2t^\alpha/L^2\ll 1,
\end{align*}
for some constants $A>0$ and $p\in\R$. We are not aware of the values of $A$ and $p$ for RLfBm or fBm. However, if $\tau_{\text{Bm}}$ denotes the FPT in \eqref{eq:tau1} for the normal diffusion case of $\alpha=1$, then it is well-known that
\begin{align*}
    \P(\tau_{\text{Bm}}\le t)
    =\textup{erf}\bigg(\sqrt{\frac{L^2}{4Dt}}\bigg)
    \sim At^pe^{-C/t}\quad\text{as }4Dt/L^2\to0,
\end{align*}
where $A=\sqrt{4D/(\pi L^2)}$, $p=1/2$, and $C=L^2/(4D)$. Therefore, if $\tau_{\text{sBm}}$ denotes the FPT in \eqref{eq:tau1} for sBm for a general $\alpha>0$, then as $4D\tc(t/\tc)^\alpha/L^2\to0$,
\begin{align}
\begin{split}\label{eq:moredetailed}
    &\P(\tau_{\text{sBm}}\le t)
    =\textup{erf}\bigg(\sqrt{\frac{L^2}{4D\tc(t/\tc)^\alpha}}\bigg)
    \sim At^p e^{-C/t^\alpha},
\end{split}    
\end{align}
where
\begin{align}\label{eq:sBmvalues}
    A=\tc^{1/2-\alpha/2}\sqrt{\frac{4D}{\pi L^2}},\quad p=\frac{\alpha}{2},\quad C=\tc^{\alpha-1}\frac{L^2}{4D}.
\end{align}

The following theorem uses the linear asymptotics in \eqref{eq:moredetailed} (rather than the mere logarithmic asymptotic in \eqref{eq:key}) to obtain the full probability distribution and higher order moment asymptotics for the fFPT $T_N$. The proof is identical to the proof of Theorem~9 in \cite{lawley2020sub}.

\begin{theorem}\label{thm:gumbel}
Let $\{\tau_{n}\}_{n\ge1}$ be a sequence of iid realizations of any random variable $\tau$ satisfying
\begin{align*}
\P(\tau\le t)
\sim At^{p}e^{-C/t^{\alpha}}\quad\text{as }t\to0^+,
\end{align*}
for some constants $C>0$, $A>0$, $\alpha>0$, and $p\in\R$. If $T_{N}:=\min\{\tau_{1},\dots,\tau_{N}\}$, then for all $x\in\R$,
\begin{align}\label{cd}
    \P(T_N>a_Nx+b_N)
    \to\exp(-e^x)\quad\text{as }N\to\infty,
\end{align}
where 
\begin{align}\label{abab}
\begin{split}
a_{N}
&=\frac{b_{N}}{\alpha\ln(AN)},\quad
b_{N}
=\Big(\frac{C}{\ln(AN)}\Big)^{1/\alpha},\quad\text{if }p=0,\\
a_{N}
&=\frac{b_{N}}{p(1+W)},\quad
b_{N}
=\Big(\frac{C\alpha}{pW}\Big)^{1/\alpha},\quad\text{if }p\neq0,
\end{split}
\end{align}
and
\begin{align}\label{dubdub}
W
&=
\begin{cases}
W_{0}\big((C\alpha/p)(AN)^{\alpha/p}\big) & \text{if }p>0,\\
W_{-1}\big((C\alpha/p)(AN)^{\alpha/p}\big) & \text{if }p<0,
\end{cases}
\end{align}
where $W_{0}(z)$ is the principal branch of the LambertW function and $W_{-1}(z)$ is the lower branch \cite{corless1996}.
\end{theorem}

We conclude this section with several comments about Theorem~\ref{thm:uni}. First, in light of the sBm asymptotics in \eqref{eq:moredetailed}-\eqref{eq:sBmvalues}, Theorem~\ref{thm:gumbel} applies to fFPTs of sBm (see section~\ref{sec:numericalsbm} and Figures~\ref{fig:SBMcd} and \ref{fig:SBM}).  

Second, the convergence in distribution in \eqref{cd} means that the distribution of the fFPT is approximately Gumbel with shape $b_{N}$ and scale $a_{N}$ in \eqref{abab},
\begin{align*}
\P(T_{N}>t)
\approx\exp\Big[-\exp\Big(\frac{t-b_{N}}{a_{N}}\Big)\Big]\quad\text{if }N\gg1.
\end{align*}

Third, the choice of $(a_N,b_N)$ in \eqref{abab} is not unique. Indeed, $(a_N,b_N)$ can be replaced by any $(a_N',b_N')$ satisfying
\begin{align*}
    \lim_{N\to\infty}a_N'/a_N=1,\quad \lim_{N\to\infty}(b_N-b_N')/a_N=0.
\end{align*}
The asymptotics of the LambertW function \cite{corless1996} thus imply that \eqref{cd} also holds for
\begin{align}
\begin{split}\label{ababuf}
a_{N}
&=\frac{C^{1/\alpha}}{\alpha(\ln N)^{1+1/\alpha}},\\
b_{N}
&=\Big(\frac{C}{\ln N}
+\frac{Cp\ln(\ln(N))}{\alpha(\ln N)^{2}}
-\frac{C\ln(AC^{p/\alpha})}{(\ln N)^{2}}\Big)^{1/\alpha}.
\end{split}
\end{align}

Finally, \eqref{cd} can be used to give explicit higher order approximations to all the moments of $T_{N}$ since essentially all of the statistical information about a Gumbel distribution is immediately available. For example, if \eqref{eq:conditiona} holds, then the mean and variance satisfy
\begin{align}
\E[T_{N}]
&=b_{N}-\gamma a_{N}+o(a_{N})\quad\text{as }N\to\infty,\label{eq:meang}\\
\textup{Var}(T_{N})
&=\frac{\pi^{2}}{6}a_{N}^{2}+o(a_{N}^{2})\quad\text{as }N\to\infty,\nonumber
\end{align}
where $\gamma\approx0.5772$ is the Euler-Mascheroni constant. Indeed, \eqref{ababuf} implies that
\begin{align*}
&\E[T_{N}]
=\Big(\frac{C}{\ln N}\Big)^{1/\alpha}\bigg\{\Big(1
+\frac{p\ln(\ln(N))}{\alpha\ln N}
-\frac{C\ln(AC^{p/\alpha})}{\ln N}\Big)^{1/\alpha}\\
&\qquad\qquad\qquad\qquad\qquad\qquad\qquad-\frac{\gamma }{\alpha\ln N}\bigg\}+\text{h.o.t.},
\end{align*}
where h.o.t.\ denotes higher order terms as $N\to\infty$. 
Notice that the leading order term in this expression agrees with Theorem~\ref{thm:uni}. The corrections give the higher order effects of $A$ and $p$.

\section{\label{sec:finite}Bounded speed anomalous diffusions}

We now consider the fFPTs of some anomalous diffusions which move at bounded speed. These processes are bounded speed analogs of the sBM in \eqref{eq:sbm} and the RLfBm in \eqref{eq:RLfBm}. These processes rely on the following bounded speed analog of a standard Brownian motion. For a dimensionless parameter $\eps>0$ and a rate $\lambda_0>0$, define
\begin{align}\label{eq:Weps}
    W_\eps(t)
    =\frac{\eps}{\sqrt{\lambda_0}}\int_0^{\lambda_0 t/\eps^2} J(s)\,\dd s,\quad t\ge0,
\end{align}
where $J(t)\in\{-1,1\}$ is a two-state, continuous-time Markov jump process that switches at unit rate,
\begin{align*}
    -1\Markov{1}{1}1.
\end{align*}
Note that $W_\eps$ has dimension $\sqrt{\mathrm{time}}$, which is the same as standard Brownian motion. 
The process $W_\eps$ in \eqref{eq:Weps} is a run-and-tumble process \cite{angelani2014first} (also called a velocity jump process \cite{dorsogna2026mean}, a dichotomous noise process \cite{bena2006}, a stochastic hybrid system \cite{hespanha2004stochastic}, or a piecewise-deterministic Markov process (PDMP) \cite{davis1984piecewise}). The process moves at a constant speed (with dimension $1/\sqrt{\mathrm{time}}$),
\begin{align*}
    \Big|\frac{\dd W_\eps(t)}{\dd t}\Big|
    =\eps^{-1}\sqrt{\lambda_0}|J(\lambda_0 t/\eps^2)|
    =\eps^{-1}\sqrt{\lambda_0}>0,
\end{align*}
and switches direction at rate
\begin{align*}
    \lambda:=\lambda_0/\eps^2.
\end{align*}
Using that $\E[J(s)J(s')]=e^{-2|s-s'|}$, the mean-squared displacement of $W_\eps$ is
\begin{align}
    \E[(W_\eps(t))^2]
    &=\frac{\eps^2}{\lambda_0}\int_0^{\lambda_0 t/\eps^2}\int_0^{\lambda_0 t/\eps^2} \E[J(s)J(s')]\,\dd s'\,\dd s\nonumber\\
    &=t-\frac{\eps^2}{2\lambda_0}(1-e^{-2\lambda_0 t/\eps^2}),\quad t\ge0.\label{eq:msdWeps} 
\end{align}
It is well-known that $W_\eps$ converges to a standard Brownian motion $W$ as $\eps$ vanishes \cite{billingsley2013},
\begin{align}\label{eq:donsker}
    W_\eps\to W\quad\text{as }\eps\to0.
\end{align}

\subsection{Bounded speed sBm}

We start with a bounded speed analog to the sBm in \eqref{eq:sbm}, defined by
\begin{align}\label{eq:finitesbm}
    X_\eps(t)
    =\sqrt{2D}\,W_\eps(\tc(t/\tc)^\alpha),
\end{align}
where $\alpha>0$ and $\eps>0$ are dimensionless, $\tc>0$ is some characteristic timescale, $W_\eps$ is in \eqref{eq:Weps}, and the diffusivity parameter $D$ compares a speed $v_0$ to the rate $\lambda_0$,
\begin{align}\label{eq:Dratio}
    D
    =\frac{v_0^2}{2\lambda_0}
    =\frac{v^2}{2\lambda}>0,
\end{align}
where
\begin{align*}
    v:=v_0/\eps,\quad \lambda:=\lambda_0/\eps^2.
\end{align*}
Put another way, if $Y_\eps(t)$ is a run-and-tumble process that moves at constant speed $v=v_0/\eps$ and switches direction at rate $\lambda=\lambda_0/\eps^2$, then $X_\eps$ in \eqref{eq:finitesbm} is the time-change $X_\eps(t)=Y_\eps(\tc(t/\tc)^\alpha)$. The time-dependent speed of $X_\eps$ is (using the relation $D=v_0^2/(2\lambda_0)$)
\begin{align*}
    \Big|\frac{\dd X_\eps(t)}{\dd t}\Big|
    &=\sqrt{2D}\Big|\frac{\dd W_\eps(\tc(t/\tc)^\alpha)}{\dd t}\Big|
    =\frac{\alpha v_0}{\eps}\Big(\frac{t}\tc\Big)^{\alpha-1}.
\end{align*}
It follows from \eqref{eq:msdWeps} that the mean-squared displacement of $X_\eps(t)$ is
\begin{align}\label{eq:msd7}
    \E[(X_\eps(t))^2]
    =2D\Big(\tc(t/\tc)^\alpha-\frac{\eps^2}{2\lambda_0}\big(1-e^{-2(\lambda_0/\eps^2)\tc(t/\tc)^\alpha}\big)\Big).
\end{align}
The convergence in \eqref{eq:donsker} implies that $X_\eps$ in \eqref{eq:finitesbm} becomes the sBm in \eqref{eq:sbm} as $\eps$ vanishes,
\begin{align}\label{eq:donskersbm}
    X_\eps
    \to X_{\text{sBm}}\quad\text{as }\eps\to0^+.
\end{align}

\subsection{Bounded speed RLfBm}

We now consider a bounded speed analog to the RLfBm in \eqref{eq:RLfBm}. These processes were introduced in Ref.~\cite{dean2021position} and their fFPTs were recently studied numerically in Ref.~\cite{grebenkov2026fastest}. For $\alpha>0$, define $X_\eps=\{X_\eps(t)\}_{t\ge0}$ by
\begin{align}\label{eq:finitespeedprocess}
    X_\eps(t)
    &=\sqrt{\alpha}\sqrt{2 D\tc^{1-\alpha}}\int_0^t (t-s)^{\alpha/2-1/2}\,\dd W_\eps(s),
\end{align}
where the diffusivity parameter $D$ is in \eqref{eq:Dratio}, $\tc>0$ is a characteristic timescale, and $W_\eps$ is in \eqref{eq:Weps}. 

The process $X_\eps$ is often multiplied by $1/\Gamma((\alpha+1)/2)$ \cite{dean2021position, grebenkov2026fastest}, but we omit this factor for simplicity. If $\alpha=1$, then the processes in \eqref{eq:finitesbm} and \eqref{eq:finitespeedprocess} both reduce to standard run-and-tumble processes that move at constant speed $v=v_0/\eps$ and switch directions at rate $\lambda=\lambda_0/\eps^2>0$. The mean-squared displacement of $X_\eps$ in \eqref{eq:finitespeedprocess} has been determined exactly in terms of the hypergeometric function \cite{dean2021position}, and at large time, it is given by
\cite{dean2021position}
\begin{align}\label{eq:msd8}
    \E[(X_\eps(t))^2]
    \sim 2D\tc(t/\tc)^\alpha\quad\text{as }t\to\infty.
\end{align}
The convergence in \eqref{eq:donsker} implies that $X_\eps$ in \eqref{eq:finitespeedprocess} becomes the RLfBm in \eqref{eq:RLfBm} as $\eps$ vanishes,
\begin{align}\label{eq:donskerrlfbm}
    X_\eps
    \to X_{\text{RLfBm}}\quad\text{as }\eps\to0^+.
\end{align}

\subsection{\label{sec:finitespeedfpts}FPTs}

If $X_\eps$ is the bounded speed process in either \eqref{eq:finitesbm} or \eqref{eq:finitespeedprocess}, consider the FPT to some positive length $L>0$,
\begin{align}\label{eq:taueps}
    \tau_\eps
    =\inf\{t\ge0:X_\eps(t)\ge L\},
\end{align}
and define the fFPT,
\begin{align}\label{eq:fFPTfinitesbm}
    T_{\eps,N}
    =\min\{\tau_{\eps,1},\dots,\tau_{\eps,N}\},
\end{align}
where $\{\tau_{\eps,n}\}_{n\ge1}$ is an iid sequence of realizations of $\tau_\eps$.

In light of \eqref{eq:donskersbm} and \eqref{eq:donskerrlfbm}, if we first take $\eps\to0$, and then take $N\to\infty$, then the analysis of section~\ref{sec:gaussian} implies that the fFPT moments decay according to
\begin{align}\label{eq:order}
    \lim_{N\to\infty}(\ln N)^{m/\alpha}\lim_{\eps\to0}\E[(T_{\eps,N})^m]=\bigg(\frac{L^2}{4D\tc^{1-\alpha}}\bigg)^{m/\alpha}.
\end{align}
However, since the speed of $X_\eps$ is bounded, the FPT is bounded away from zero. In particular,
\begin{align}\label{eq:lowerbound}
    T_{\eps,N}\ge\tmin>0,
\end{align}
for some minimal time $\tmin>0$, which contradicts \eqref{eq:order}. The minimal time is obtained by solving $X_\eps(\tmin)=L$ for $\tmin$ in the case $J(t)=1$ for all $t\in[0,\tmin]$. Hence, if we take $N\to\infty$ for fixed $\eps>0$, then the fFPT converges to 
\begin{align}\label{eq:epsfixed}
    \lim_{N\to\infty}T_{\eps,N}
    =\tmin\quad\text{almost surely if }\eps>0.
\end{align}
We now investigate the convergence in \eqref{eq:epsfixed}. 

\subsection{\label{sec:finitespeedsbmfpt}fFPT for bounded speed sBm}

For the bounded speed sBm in \eqref{eq:finitesbm}, the smallest that the FPT $\tau_\eps$ in \eqref{eq:taueps} can be is
\begin{align}\label{eq:tmin}
    \tmin
    :=\tc\Big(\frac{\eps L}{v_0\tc}\Big)^{1/\alpha}
    =\tc\Big(\frac{L}{v\tc}\Big)^{1/\alpha}.
\end{align}
Notice that the minimal search time $\tmin$ in \eqref{eq:tmin} is (a) an increasing function of $\alpha$ if $\tc v/L>1$ and (b) a decreasing function of $\alpha$ if $\tc v/L<1$. 

If $\sigma_\eps$ denotes the FPT in \eqref{eq:taueps} for the case $\alpha=1$, then it was shown in section~4.1 in Ref.~\cite{lawley2021pdmp} that
\begin{align*}
    \P(\sigma_\eps=L/v)
    &=q,\\
    \P(L/v<\sigma_{\eps}<(L/v)(1+\delta))
    &\sim(1-q)A_1\delta\quad\text{as }\delta\to0^+,
\end{align*}
where
\begin{align}\label{eq:qAsBm}
    q
    :=p_1e^{-\kappa}\in(0,1),\quad
    A_1
    :=\frac{\kappa(p_0+\kappa p_1)}{2(e^\kappa-p_1)}>0,
\end{align}
where $\P(J(0)=1)=p_1=1-p_0\in(0,1]$, and
\begin{align}\label{eq:kappa}
    \kappa
    =\frac{\lambda L}{v}
    =\frac{\lambda_0 L}{\eps v_0}>0.
\end{align}
Note that $q$ is the probability that the process starts in the positive direction and does not change direction before hitting $L$.

For the general case of $\alpha>0$, it follows that 
\begin{align*}
    \P(\tmin<\tau_{\eps}<\tmin(1+\delta))
    &=\P(L/v<\sigma_{\eps}<(L/v)(1+\delta)^\alpha)\\
    &\sim(1-q)A\delta\quad\text{as }\delta\to0^+,
\end{align*}
where $q$ is in \eqref{eq:qAsBm} and
\begin{align}\label{eq:A}
    A
    :=\alpha A_1
    =\alpha\frac{\kappa(p_0+\kappa p_1)}{2(e^\kappa-p_1)}>0.
\end{align}
We can then apply the theory of Ref.~\cite{lawley2021pdmp} to determine the large $N$ distribution and moments of the fFPT $T_{\eps,N}$.

Specifically, Theorem~2 in Ref.~\cite{lawley2021pdmp} implies
\begin{align*}
T_{\eps,N}
=\begin{cases}
    \tmin & \text{with prob.\ }1-(1-
    q)^N,\\
    \tmin(1+\psi_N/(AN)) & \text{with prob.\ }(1-q)^N,\\
\end{cases}
\end{align*}
where $A$ is in \eqref{eq:A} and $\psi_{N}>0$ converges in distribution to a unit mean exponential random variable, which means that for all $x\ge0$,
\begin{align}\label{eq:cdpdmp1}
    \P(\psi_N>x)
    \to e^{-x}\quad\text{as }N\to\infty.
\end{align}
Furthermore, it is known that $\P(\tau_{\eps}>t)$ decays as $t^{-1/2}$ as $t\to\infty$ if $\alpha=1$ \cite{malakar2018steady}, and therefore $\E[T_{\eps,N}]<\infty$ for all $N>2/\alpha$. Theorem~4 in Ref.~\cite{lawley2021pdmp} thus yields the following asymptotics for each moment $m\ge1$, 
\begin{align*}
\E[(T_{\eps,N}-\tmin)^{m}]
&\sim
(1-q)^{N}m!\Big(\frac{\tmin}{AN}\Big)^m\quad\text{as }N\to\infty.
\end{align*}
Hence, the mean and variance satisfy
\begin{align}
\E[T_{\eps,N}]
&= \tmin\bigg(1+\frac{(1-q)^{N}}{AN}\bigg)+\text{h.o.t.},\label{eq:mean1}\\
\textup{Var}(T_{\eps,N})
&=2(1-q)^N\Big(\frac{\tmin}{AN}\Big)^2+\text{h.o.t.},\nonumber
\end{align}
where h.o.t.\ denotes higher order terms as $N\to\infty$.

\subsection{\label{sec:finitespeedrlfbmfpt}fFPT for bounded speed RLfBm}

For the bounded speed RLfBm in \eqref{eq:finitespeedprocess}, the smallest that the FPT $\tau_\eps$ can be is
\begin{align}\label{eq:tminrlfbm}
    \tmin
    :=\tc\Big[\Big(\frac{\beta}{\sqrt{\alpha}}\Big)\Big(\frac{\eps L}{v_0\tc}\Big)\Big]^{1/\beta}
    =\tc\Big[\Big(\frac{\beta}{\sqrt{\alpha}}\Big)\Big(\frac{L}{v\tc}\Big)\Big]^{1/\beta},
\end{align}
where
\begin{align*}
    \beta
    =\alpha/2+1/2
    =(\alpha+1)/2.    
\end{align*}
A simple calculus and algebra exercise shows that the minimal search time $\tmin$ in \eqref{eq:tminrlfbm} is (a) an increasing function of $\alpha$ if $\tc v/L>\chi$ and (b) a decreasing function of $\alpha$ if $\tc v/L<\chi$, where
\begin{align*}
    \chi
    =\frac{2\sqrt{\alpha}}{1+\alpha}\exp\Big(\frac{\alpha-1}{2\alpha}\Big).
\end{align*}

It is immediate that 
\begin{align}\label{eq:qforrlfbm}
    \P(\tau_\eps=\tmin)
    =q
    :=p_1e^{-\lambda\tmin}\in(0,1),
\end{align}
which, again, is the probability that the process starts in the positive direction and does not change direction before hitting $L$. 
We show in Appendix~\ref{sec:shorttimeRLfbm} that
\begin{align*}
    \P(\tmin<\tau_{\eps}<\tmin(1+\delta))
    &\sim(1-q)A\delta^p\quad\text{as }\delta\to0^+,
\end{align*}
where the power $p\in(0,1]$ and prefactor $A$ are
\begin{align}\label{eq:pA}
\begin{split}
    p
    &=\begin{cases}
        1 &\text{if }\alpha\in(0,1],\\
        1/\beta &\text{if }\alpha>1,
    \end{cases}\\
    A
    &=\begin{cases}
        \displaystyle\Big(p_0+p_1\frac{\lambda \tmin}{2-\beta}\Big)\frac{\lambda\tmin}{2}\frac{e^{-\lambda\tmin}}{1-p_1e^{-\lambda\tmin}} &\text{if }\alpha\in(0,1],\\
        \displaystyle\frac{p_1\lambda \tmin e^{-\lambda \tmin}}{1-p_1 e^{-\lambda\tmin}}\Big(\frac{\beta}{2}\Big)^{1/\beta} &\text{if }\alpha>1.
    \end{cases}
\end{split}    
\end{align}
We can then apply the theory of Ref.~\cite{lawley2021pdmp} to determine the large $N$ distribution and moments of the fFPT $T_{\eps,N}$.

Specifically, Theorem~2 in Ref.~\cite{lawley2021pdmp} implies
\begin{align*}
T_{\eps,N}
=\begin{cases}
    \tmin & \text{with prob.\ }1-(1-
    q)^N,\\
    \displaystyle\tmin\Big(1+\frac{\psi_N}{(AN)^{1/p}}\Big) & \text{with prob.\ }(1-q)^N,\\
\end{cases}
\end{align*}
where $p\in\{1,1/\beta\}$ and $A$ are in \eqref{eq:pA} and $\psi_{N}>0$ converges in distribution to a Weibull random variable with unit scale and shape $p$, which means that for all $x\ge0$,
\begin{align}\label{eq:cdpdmp2}
    \P(\psi_N>x)
    \to e^{-x^p}\quad\text{as }N\to\infty.
\end{align}
Furthermore, if we assume that $\E[T_{\eps,N}]<\infty$ for some $N\ge1$, then Theorem~4 in Ref.~\cite{lawley2021pdmp} yields the following asymptotics for each moment $m\ge1$ as $N\to\infty$,
\begin{align*}
\E[(T_{\eps,N}-\tmin)^{m}]
&\sim(1-q)^{N}\Gamma(1+m/p)\Big(\frac{\tmin}{(AN)^{1/p}}\Big)^{m}.
\end{align*}
Hence, the mean and variance satisfy
\begin{align}
\E[T_{\eps,N}]
&= \tmin\bigg(1+\frac{(1-q)^{N}}{(AN)^{1/p}}\Gamma\Big(1+\frac{1}{p}\Big)\bigg)+\text{h.o.t.},\label{eq:mean2}\\
\textup{Var}(T_{\eps,N})
&=(1-q)^N\Gamma\Big(1+\frac{2}{p}\Big)\Big(\frac{\tmin}{(AN)^{1/p}}\Big)^2+\text{h.o.t.},\nonumber
\end{align}
where h.o.t.\ denotes higher order terms as $N\to\infty$.

\subsection{\label{sec:crossover}Estimating the crossover}

For the bounded speed processes $X_\eps$ in \eqref{eq:finitesbm} and \eqref{eq:finitespeedprocess}, the analysis in section~\ref{sec:gaussian} and the analysis in section~\ref{sec:finite} offer differing estimates for the fFPT $T_{\eps,N}$. For instance, section~\ref{sec:gaussian} suggests that the mean fFPT for $N\gg1$ is approximately,
\begin{align}\label{eq:approx1}
    \E[T_{\eps,N}]
    \approx\tc\bigg(\frac{L^2}{4D\tc\ln N}\bigg)^{1/\alpha},
\end{align}
whereas section~\ref{sec:finite} suggests
\begin{align}\label{eq:approx2}
    \E[T_{\eps,N}]
    \approx\tmin\bigg(1+\frac{(1-q)^{N}}{(AN)^{1/p}}\Gamma(1+1/p)\bigg),
\end{align}
where $q$, $A$, and $p$ are given in sections~\ref{sec:finitespeedsbmfpt}-\ref{sec:finitespeedrlfbmfpt}.

In order to estimate when the section~\ref{sec:gaussian} versus section~\ref{sec:finite} approximations are accurate, observe that the lower bound in \eqref{eq:lowerbound} implies that the accuracy of the estimate in \eqref{eq:approx1} requires
\begin{align*}
    \tc\bigg(\frac{L^2}{4D\tc\ln N}\bigg)^{1/\alpha}
    \gg\tmin,
\end{align*}
which is equivalent to
\begin{align*}
    N
    \ll\exp\Big[\frac{L^2}{4D\tc}\Big(\frac{\tc}{\tmin}\Big)^\alpha\Big]
    =\exp\Big[\frac{1}{2}\frac{\lambda L}{v}\frac{L}{v\tc}\Big(\frac{\tc}{\tmin}\Big)^\alpha\Big],
\end{align*}
where we used $D=v^2/(2\lambda)$. 
We therefore predict that the approximations of section~\ref{sec:gaussian} are accurate if 
\begin{align}\label{eq:suff}
    1
    \ll N
    \ll N_1
    :=\exp\Big[\frac{1}{2}\frac{\lambda L}{v}\frac{L}{v\tc}\Big(\frac{\tc}{\tmin}\Big)^\alpha\Big],
\end{align}
where the equality defines $N_1$. Furthermore, we predict that the approximations of section~\ref{sec:finite} are accurate if the leading order correction term to the mean in \eqref{eq:approx2} is small, which means
\begin{align*}
    N
    \gg N_2:=1/q
    =e^{\lambda\tmin}/p_1,
\end{align*}
where the equality defines $N_2$ and $q=\P(\tau_\eps=\tmin)=p_1e^{-\lambda\tmin}$ with $p_1=\P(J(0)=1)$.

For the bounded speed sBm $X_\eps$ in \eqref{eq:finitesbm}, the value of $\tmin$ in \eqref{eq:tmin} implies that $N_1$ is
\begin{align}\label{eq:N1sbm}
    N_{1}
    =\exp(\kappa/2),
\end{align}
where $\kappa=\lambda L/v=\eps^{-1}\lambda_0 L/v_0$ as in \eqref{eq:kappa}, and the value of $q$ in \eqref{eq:qAsBm} implies that $N_2$ is 
\begin{align}\label{eq:N2sbm}
    N_2
    =e^\kappa/p_1=(N_{1})^2/p_1.
\end{align}

For the bounded speed RLfBm $X_\eps$ in \eqref{eq:finitespeedprocess}, the value of $\tmin$ in \eqref{eq:tminrlfbm} implies that $N_1$ is
\begin{align*}
    N_{1}
    &=\exp\Big[\frac{\kappa}{2}\Big(\frac{\sqrt{\alpha}}{\beta}\Big)^{\alpha/\beta}\Big(\frac{L}{v\tc}\Big)^{(1-\alpha)/(1+\alpha)}\Big],
\end{align*}
and the value of $q$ in \eqref{eq:qforrlfbm} implies that $N_2$ is 
\begin{align*}
    N_2
    =(N_{1})^{(1+\alpha)^2/(2\alpha)}/p_1.
\end{align*}

Finally, since the mean fFPT $\E[T_{\eps,N}]$ is a monotonically decreasing function of $N$, a simple approximation for all $N\gg1$ is the maximum of the unbounded speed theory in \eqref{eq:approx1} and the leading order bounded speed theory in \eqref{eq:approx2}, i.e.
\begin{align*}
    \E[T_{\eps,N}]
    \approx\max\bigg\{\tc\bigg(\frac{L^2}{4D\tc\ln N}\bigg)^{1/\alpha},\tmin\bigg\}\quad\text{if }N\gg1.
\end{align*}

\section{\label{sec:numerics}Numerical simulations}

We now illustrate some of our results using numerical simulations. Throughout this section, we take
\begin{align*}
    D=v_0=L=\tc=1,\quad \lambda_0=1/2.
\end{align*}
The simulation details are in Appendix~\ref{sec:numericaldetails}.

\subsection{\label{sec:numericalsbm}sBm}

We start with the sBm in \eqref{eq:sbm} and its bounded speed analog in \eqref{eq:finitesbm}.

Figure~\ref{fig:SBMcd} plots the convergence in distribution implied by Theorem~\ref{thm:gumbel} for $\alpha=1$, $\alpha=0.75$, and $\alpha=1.5$ in the left, middle, and right plots, respectively. In particular, the solid black curve is the limiting survival probability $\exp(-e^x)$, and the circle markers are the distribution of $(T_N-b_N)/a_N$ for sBm where $(a_N,b_N)$ are given in Theorem~\ref{thm:gumbel} in terms of the LambertW function. Furthermore, the plus markers are the distribution of $(T_{\eps,N}-b_N)/a_N$, where $T_{\eps,N}$ is the fFPT in \eqref{eq:fFPTfinitesbm} of the bounded speed process in \eqref{eq:finitesbm} with $\eps=10^{-3}$. This plot shows excellent agreement between theory and simulations for both the unbounded speed sBm and its bounded speed analog.

\begin{figure*}
\centering
\includegraphics[width=1\linewidth]{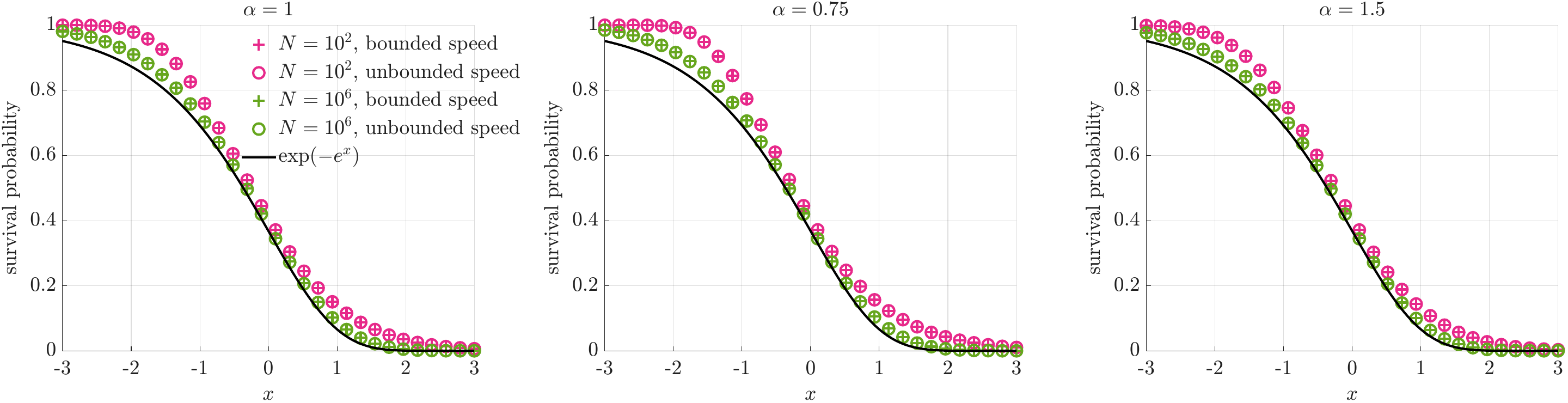}
\caption{Convergence in distribution of fFPTs of the unbounded speed sBm in \eqref{eq:sbm} and its bounded speed analog in \eqref{eq:finitesbm}.}
\label{fig:SBMcd}
\end{figure*}

\begin{figure*}
\centering
\includegraphics[width=.9\linewidth]{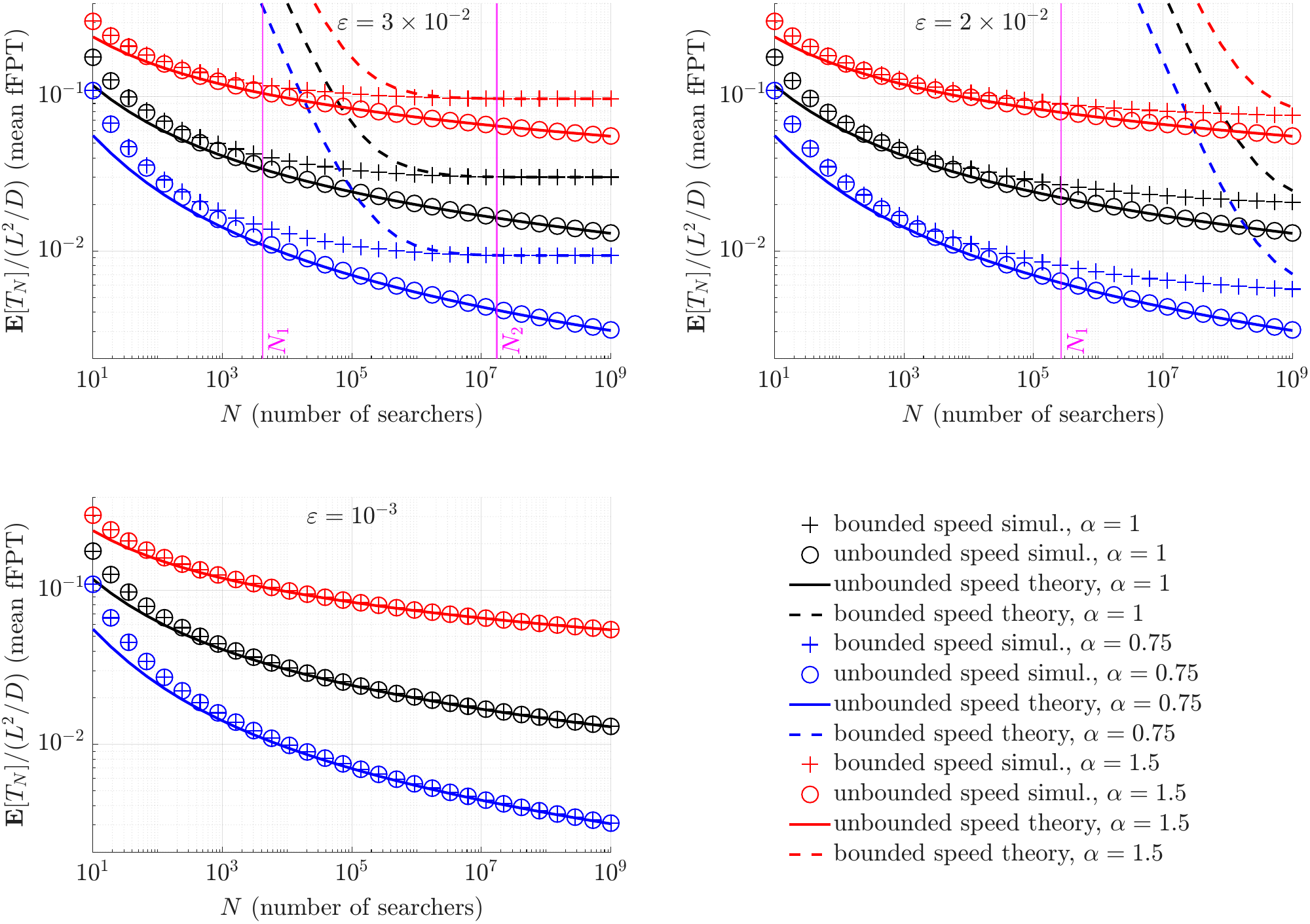}
\caption{Decay of the mean fFPT as $N$ increases for the unbounded speed sBm in \eqref{eq:sbm} (circle markers) and the bounded speed sBm in \eqref{eq:finitesbm} (plus markers) for different values of $\eps$ in the three panels (the circle markers are the same in each panel). The three colors correspond to subdiffusion (blue), normal diffusion (black), and superdiffusion (red). The solid curves are the estimates to the mean fFPT given in \eqref{eq:meang} where $a_N$ and $b_N$ are in Theorem~\ref{thm:gumbel} with $A$, $p$, and $C$ in \eqref{eq:sBmvalues}. The dashed curves are the mean estimate in \eqref{eq:mean1}. The pink vertical bars are $N_1$ in \eqref{eq:N1sbm} and $N_2$ in \eqref{eq:N2sbm}.}
\label{fig:SBM}
\end{figure*}

\begin{figure*}
\centering
\includegraphics[width=1\linewidth]{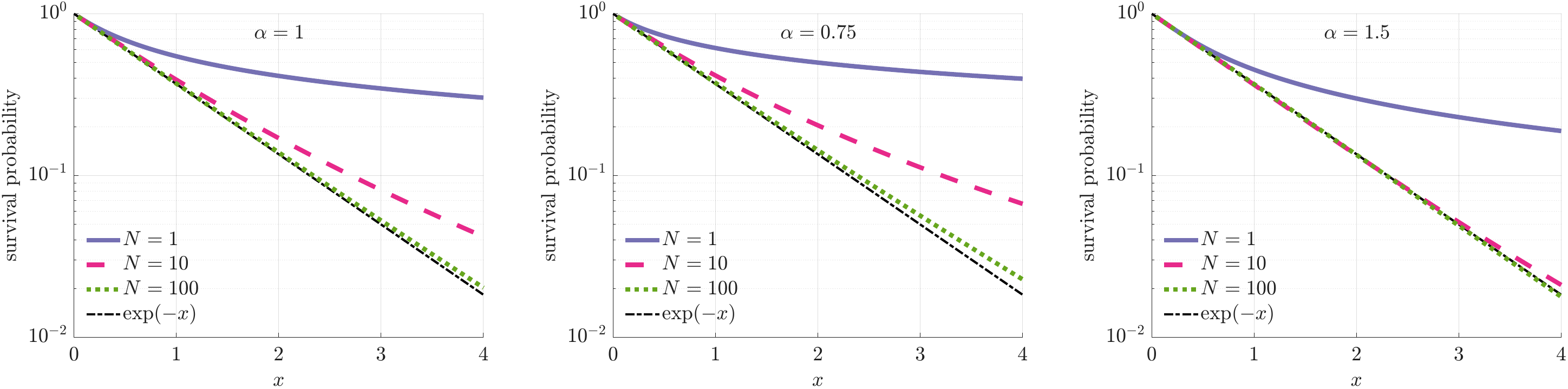}
\caption{The convergence in distribution in \eqref{eq:cdpdmp1} for $\eps=1$ and $\alpha=1$ (left), $\alpha=0.75$ (middle), and $\alpha=1.5$ (right).}
\label{fig:SBMpdmp}
\end{figure*}

\begin{figure}
\centering
\includegraphics[width=1\linewidth]{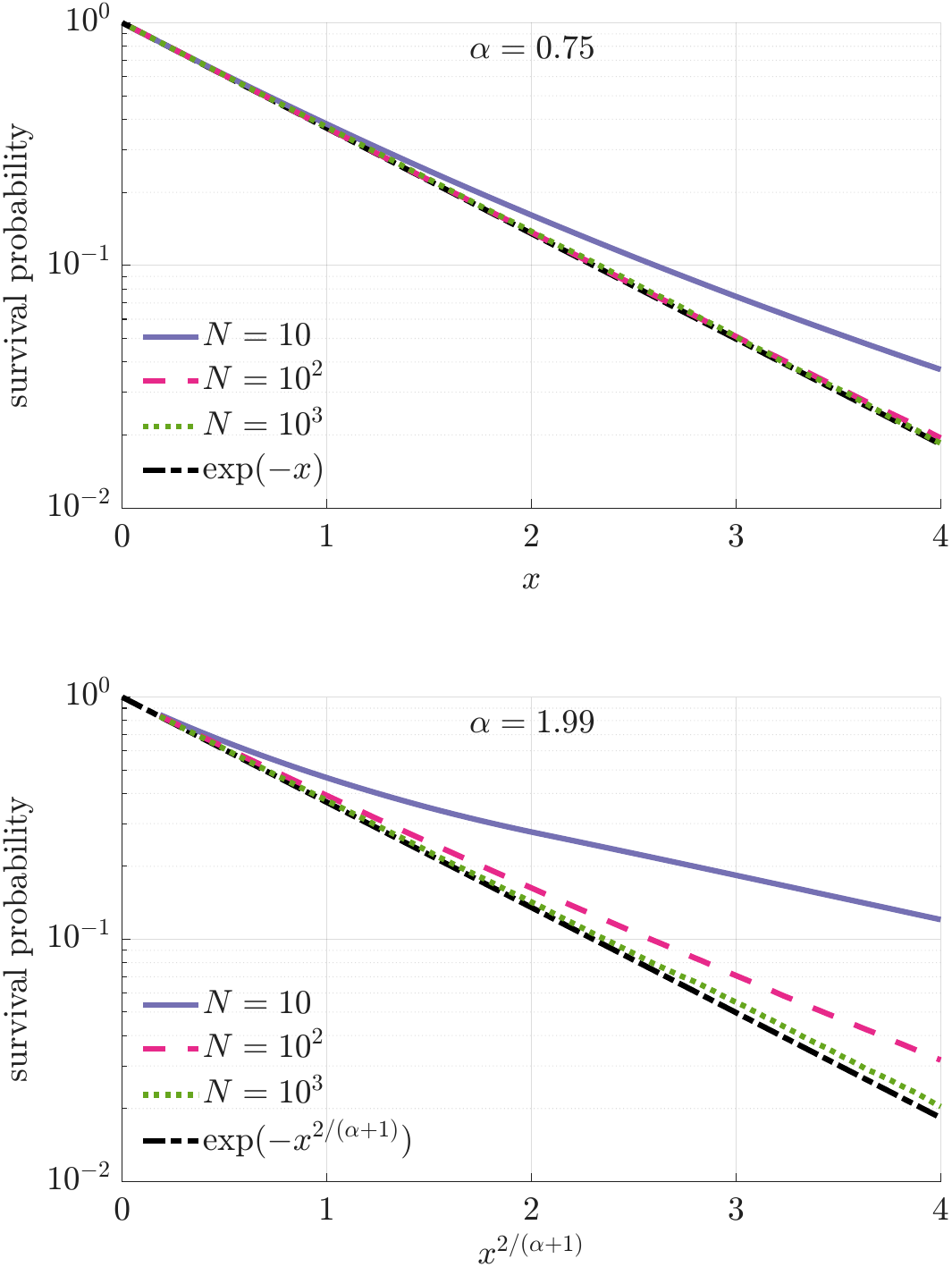}
\caption{The convergence in distribution in \eqref{eq:cdpdmp2} for $\eps=1$ and $\alpha=0.75$ (top panel) and $\alpha=1.99$ (bottom panel).}
\label{fig:FBMpdmp}
\end{figure}

Figure~\ref{fig:SBM} plots the decay of the mean fFPT as $N$ increases for the unbounded speed sBm in \eqref{eq:sbm} (circle markers) and the bounded speed sBm in \eqref{eq:finitesbm} (plus markers) for different values of $\eps$ in the three panels (the circle markers are the same in each panel). There are four important features of Figure~\ref{fig:SBM} which we emphasize.

First, the subdiffusive mean fFPT is the fastest (blue), followed by normal diffusion (black), then superdiffusion (red). Notice that this holds for both the unbounded speed process and its bounded speed counterpart.

Second, the mean fFPTs of the unbounded speed and bounded speed processes agree closely if $1\le N\ll N_1$, where $N_1$ is in \eqref{eq:N1sbm}. For $\eps=10^{-3}\ll1$ in the bottom left panel, we observe near perfect agreement for the entire range of $N$ plotted since $N_1>10^{108}$ (the bounded speed theory is not visible in this bottom left panel).

Third, the circle markers (unbounded speed sBm simulations) agree closely with the theory of section~\ref{sec:gaussian} for $N\gg1$. In particular, the solid curves are the estimates to the mean fFPT given in \eqref{eq:meang} where $a_N$ and $b_N$ are in Theorem~\ref{thm:gumbel} with $A$, $p$, and $C$ in \eqref{eq:sBmvalues}. The mean estimates from Corollary~\ref{cor:decay} are similar, though the agreement is not quite as close (these curves are omitted to simplify the figure).

Fourth, in agreement with the predictions of section~\ref{sec:crossover}, the plus markers (bounded speed sBm) agree with the theory of section~\ref{sec:gaussian} for $1\ll N\ll N_1$ and agree with the theory of section~\ref{sec:finite} (dashed curves, which are the mean estimate in \eqref{eq:mean1}) for $N\gg N_2$, where $N_1$ is in \eqref{eq:N1sbm} and $N_2$ is in \eqref{eq:N2sbm}.

To more closely illustrate the theory of section~\ref{sec:finite}, Figure~\ref{fig:SBMpdmp} plots the convergence in distribution in \eqref{eq:cdpdmp1} for $\eps=1$ and $\alpha=1$ (left panel), $\alpha=0.75$ (middle panel), and $\alpha=1.5$ (right panel).

\subsection{RLfBm}

We now consider the RLfBm in \eqref{eq:RLfBm} and its bounded speed analog in \eqref{eq:finitespeedprocess}. Figure~\ref{fig:FBMpdmp} plots the convergence in distribution in \eqref{eq:cdpdmp2} for $\eps=1$ and $\alpha=0.75$ (top panel) and $\alpha=1.99$ (bottom panel).

Figure~\ref{fig:FBM} plots the decay of the mean fFPT as $N$ increases for the bounded speed RLfBm in \eqref{eq:finitespeedprocess} (plus markers) for $\alpha=0.75$ (blue) and $\alpha=1$ (black). The blue and black dashed curves are the estimate to the mean fFPT,
\begin{align*}
    \E[T_N]
    \approx \tc\bigg(\frac{L^2}{4D\tc\ln N}\bigg)^{1/\alpha},
\end{align*}
given by Corollary~\ref{cor:decay}. The black solid curve is the estimate to the mean fFPT given in \eqref{eq:meang} where $a_N$ and $b_N$ are in Theorem~\ref{thm:gumbel} with $A$, $p$, and $C$ in \eqref{eq:sBmvalues} for $\alpha=1$. 

We again see that the subdiffusive mean fFPT is faster than the normal diffusive mean fFPT, even for the finite speed RLfBm. We also note that the unbounded speed theory agrees closely with the bounded speed simulations. In this plot, we take $\eps=5\times10^{-3}$.

\begin{figure}
\centering
\includegraphics[width=1\linewidth]{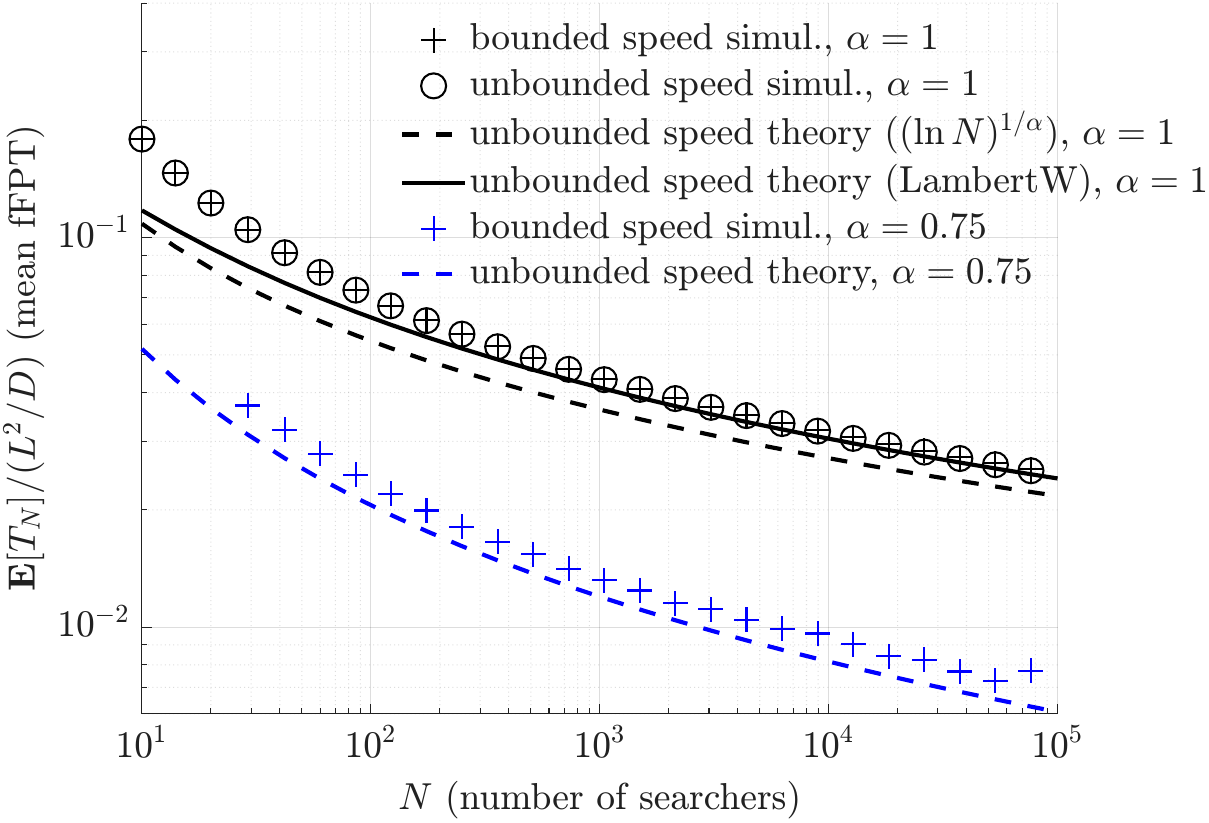}
\caption{Finite speed RLfBm in \eqref{eq:finitespeedprocess} agrees with the unbounded speed theory of section~\ref{sec:gaussian}. See the text for details.}
\label{fig:FBM}
\end{figure}

\section{Discussion}

In this paper, we studied fFPTs for broad classes of anomalous diffusion models. We found that the fFPT decays logarithmically in the number of searchers $N$,
\begin{align}\label{eq:logapprox}
    \E[(T_N)^m]
    \sim(\tc)^m\bigg(\frac{L^2}{4D\tc\ln N}\bigg)^{m/\alpha}\quad\text{as }N\to\infty.
\end{align}
We found that the logarithmic estimate in \eqref{eq:logapprox} is valid both for models with unbounded speed and models with bounded speed, though the rigorous asymptotic equivalence in \eqref{eq:logapprox} holds only for the unbounded speed models of section~\ref{sec:gaussian}. Indeed, for the bounded speed models of section~\ref{sec:finite}, the logarithmic decay eventually transitions to an exponential convergence to a strictly positive minimal search time $\tmin>0$ as $N$ increases. We characterized this exponential convergence in terms of a Weibull distribution. 

The decay in \eqref{eq:logapprox} implies the counterintuitive result that subdiffusion can be faster than normal diffusion and normal diffusion can be faster than superdiffusion. This agrees qualitatively with a prior analysis of the time fractional Fokker-Planck equation model of subdiffusion \cite{lawley2020sub}, though the dependence on the anomalous diffusion exponent $\alpha$ differs (compare $1/(\ln N)^{2/\alpha-1}$ in \eqref{eq:sub} to $1/(\ln N)^{1/\alpha}$ in \eqref{eq:logapprox}). We found that this counterintuitive result holds for both unbounded speed and bounded speed models.

To estimate the parameter regimes in which this counterintuitive result holds, observe that the counterintuitive result means that the estimate in \eqref{eq:logapprox} is an increasing function of $\alpha$, which is equivalent to 
\begin{align}\label{eq:Ncounter}
    N
    >\exp\bigg(\frac{L^2}{4D\tc}\bigg).
\end{align}
For the unbounded speed models of section~\ref{sec:gaussian}, we can always take $N$ large enough so that \eqref{eq:Ncounter} holds. For the bounded speed models of section~\ref{sec:finite}, the accuracy of the logarithmic estimate in \eqref{eq:logapprox} requires $1\ll N\ll N_1$, where
\begin{align*}
    N_1
    =\exp\Big[\frac{1}{2}\frac{\lambda L}{v}\frac{L}{v\tc}\Big(\frac{\tc}{\tmin}\Big)^\alpha\Big]
    =\exp\Big[\frac{L^2}{4D\tc}\Big(\frac{\tc}{\tmin}\Big)^\alpha\Big],
\end{align*}
where we used $D=v^2/(2\lambda)$. It is thus only possible to satisfy \eqref{eq:Ncounter} and $N\ll N_1$ if
\begin{align}\label{eq:makessense}
    \tc>\tmin.
\end{align}
For the bounded speed sBm, the formula for $\tmin$ in \eqref{eq:tmin} implies that \eqref{eq:makessense} is equivalent to
\begin{align}\label{eq:Ncountersbm}
    \tc>L/v.
\end{align} 
For the bounded speed RLfBm, the formula for $\tmin$ in  \eqref{eq:tminrlfbm} implies that \eqref{eq:makessense} is equivalent to
\begin{align}\label{eq:Ncounterrlfbm}
    \tc
    >\frac{1}{2}\Big(\sqrt{\alpha}+\frac{1}{\sqrt{\alpha}}\Big)(L/v).
\end{align}
The counterintuitive result thus requires that the characteristic time $\tc$ be sufficiently large compared to the ballistic travel time $L/v$, which is intuitive. Indeed, recall from \eqref{eq:msd1}, \eqref{eq:msd7}, and \eqref{eq:msd8} that the mean-squared displacement for these models is (approximately for the bounded speed models)
\begin{align*}
2D\tc(t/\tc)^\alpha,
\end{align*}
which is a decreasing function of $\alpha$ if $t<\tc$. Furthermore, \eqref{eq:Ncountersbm}-\eqref{eq:Ncounterrlfbm} are precisely the conditions in which $\tmin$ is an increasing function of $\alpha$. Hence, if \eqref{eq:Ncountersbm} or \eqref{eq:Ncounterrlfbm} is satisfied, then subdiffusion is faster than normal diffusion which is faster than superdiffusion for all sufficiently large $N$ for the bounded speed model.

We now comment on possible extensions our results. First, we focused on the fFPT for simplicity, but our results are readily extended to the $k$th fFPT for $1\le k\ll N$ (for instance, see Theorem~1 in Ref.~\cite{lawley2020uni}, section 3.2 in Ref.~\cite{lawley2020dist}, and section 3.2 in Ref.~\cite{lawley2021pdmp}). We also considered the FPT to $L>0$, but we could readily extend to escape from a finite interval $(-L_0,L)$. In fact, we could also consider fFPTs for search processes in higher spatial dimensions. Indeed, sections 5 and 6 in Ref.~\cite{lawley2021pdmp} considered fFPTs of run-and-tumble process in two and three spatial dimensions.

We conclude by mentioning some related work. Motivated by the non-physical predictions of models with unbounded speed, the fFPTs of bounded speed processes were studied in Ref.~\cite{lawley2021pdmp}. The particular class of models were PDMPs \cite{davis1984piecewise}, which includes run-and-tumble processes. In studying photon transport experiments, Ref.~\cite{carroll2025measurements} also studied fFPTs of bounded speed processes (in terms of the telegraph equation which describes run-and-tumble processes) and compared them to fFPTs of unbounded speed processes. Ref.~\cite{grebenkov2026fastest} recently studied fFPTs of run-and-tumble processes, including a numerical analysis of the bounded speed RLfBm in \eqref{eq:finitespeedprocess}. Refs.~\cite{hass2023anomalous, hass2024extreme, hass2024first} studied fFPTs of random walks which hop along a discrete lattice at discrete times and thus move at bounded speed. The fFPTs of discrete-time, discrete-space processes were also recently studied in Ref.~\cite{karamched2026entropic}. Ref.~\cite{lawley2020networks} studied fFPTs of continuous-time, discrete-space processes and, similar to the present work, found that the many searcher limit ($N\to\infty$) and the diffusion limit do not commute (see section~V in Ref.~\cite{lawley2020networks}). The fFPTs of superdiffusive processes, including L\'evy flights, were studied in Ref.~\cite{lawley2023super}, and it was found that the mean fFPT decays as $1/N$ as $N\to\infty$ rather than the logarithmic decay found in the present work.

\appendix

\section{\label{sec:proofs}Proofs}

We now prove Proposition~\ref{prop:key} and Theorem~\ref{thm:uni}. 

\begin{proof}[Proof of Proposition~\ref{prop:key}]
    For any $R>0$, define the FPT
    \begin{align*}
        \tau_R
        :=\inf\{t\ge0:X(t)\ge R\}.
    \end{align*}
    Large deviation theory implies (see, for example, page 59 in \cite{lifshits2012lectures}) that for any $T>0$,
    \begin{align*}
        \lim_{R\to\infty}R^{-2}\ln\P(\tau_{R}\le T)
        =-\frac{1}{4D\tc(T/\tc)^\alpha}.
    \end{align*}
    Self-similarity implies that for any $c>0$,
    \begin{align*}
        \P(\tau_R\le T)
        =\P(c\tau_{c^{-\alpha/2}R}\le T).
    \end{align*}
    Taking $R=t^{-\alpha/2}$, $T=L^{-2/\alpha}$, and $c=t^{-1}L^{-2/\alpha}$ completes the proof.
\end{proof}

Before proving Theorem~\ref{thm:uni}, we estimate an integral.

\begin{lemma}\label{lem:integralestimate}
If $\alpha>0$, $C>0$, and $\delta>0$, then 
    \begin{align*}
        \int_0^\delta (1-e^{-C/t^\alpha})^N\,\dd t
        \sim C^{1/\alpha}(\ln N)^{-1/\alpha}\quad\text{as }N\to\infty.
    \end{align*}
\end{lemma}

\begin{proof}[Proof of Lemma~\ref{lem:integralestimate}]
Since $1-e^{-C/t^\alpha}$ is a decreasing function of $t>0$, we have that for any $\delta'\in(0,\delta)$
\begin{align*}
    \int_{\delta'}^\delta (1-e^{-C/t^\alpha})^N\,\dd t
    \le(1-e^{-C/(\delta')^\alpha})^N(\delta-\delta'),
\end{align*}
which vanishes exponentially fast as $N\to\infty$. Hence, we may take $\delta>0$ as small as we like. Changing variables $s=t^\alpha/C$ yields
\begin{align*}
    \int_0^\delta (1-e^{-C/t^\alpha})^N\,\dd t
    =\frac{C^{1/\alpha}}{\alpha}\int\limits_0^{\delta^\alpha/C}s^{1/\alpha-1}(1-e^{-1/s})^N\,\dd s,
\end{align*}
and therefore it is enough to determine the large $N$ behavior of 
\begin{align*}
    I
    &:=\int_0^{\delta}s^{1/\alpha-1}(1-e^{-1/s})^N\,\dd s\\
    &=\int_0^{\delta}s^{1/\alpha-1}\exp(N\ln(1-e^{-1/s}))\,\dd s.
\end{align*}
An elementary calculation shows that
\begin{align}\label{logbounds}
-x(1+x)
\le\ln(1-x)
\le-x,\quad \text{if }x\in[0,1/2].
\end{align}
Taking $\delta$ sufficiently small so that $e^{-1/s}\le1/2$ for all $s\in(0,\delta]$ and using \eqref{logbounds} yields the bounds
\begin{align}\label{Iplus}
\begin{split}
    I_{-}
    &:=\int_0^{\delta}s^{1/\alpha-1}\exp(-Ne^{-1/s}(1+e^{-1/s}))\,\dd s
    \le I\\
    &\le\int_0^{\delta}s^{1/\alpha-1}\exp(-Ne^{-1/s}))\,\dd s=:I_{+}(N).
\end{split}
\end{align}
Furthermore, since $e^{-1/s}$ is increasing, we have the lower bound
\begin{align}\label{Iminus}
\begin{split}
I_{-}
\ge\int_0^{\delta}s^{1/\alpha-1}\exp(-Ne^{-1/s}(1+e^{-1/\delta}))\,\dd s\\
=I_{+}((1+e^{-1/\delta})N),
\end{split}
\end{align}
where $I_{+}(\cdot)$ is defined in \eqref{Iplus}. 

To study $I_{+}(AN)$ for an arbitrary $A\ge1$, we change variables $u=ANe^{-1/s}$ to obtain
\begin{align*}
    I_+(AN)
    &=\int_0^{\delta}s^{1/\alpha-1}\exp(-ANe^{-1/s}))\,\dd s\\
    &=\int_0^{ANe^{-1/\delta}}u^{-1}(\ln(ANu^{-1}))^{-1/\alpha-1}e^{-u}\,\dd u\\
    &=I_{0,1}
    +I_{1,(\ln N)/2}
    +I_{(\ln N)/2,ANe^{-1/\delta}},
\end{align*}
where 
\begin{align*}
    I_{a,b}
    :=\int_a^b u^{-1}(\ln(ANu^{-1}))^{-1/\alpha-1}e^{-u}\,\dd u.
\end{align*}
Starting with the third integral, since $u^{-1}$ is decreasing and $(\ln(ANu^{-1}))^{-1/\alpha-1}$ is increasing for $u\in[0,ANe^{-1/\delta}]$, we have that
\begin{align*}
    I_{(\ln N)/2,ANe^{-1/\delta}}
    &\le\frac{2}{\ln N}(\ln(e^{1/\delta}))^{-1/\alpha-1}\int_{(\ln N)/2}^{ANe^{-1/\delta}} e^{-u}\,\dd u\\
    &\le\frac{2}{\ln N}(\ln(e^{1/\delta}))^{-1/\alpha-1}\frac{1}{\sqrt{N}}\\
    &=o((\ln N)^{-1/\alpha})\quad\text{as }N\to\infty.
\end{align*}
Moving to the second integral and again using that $u^{-1}$ is decreasing and $(\ln(ANu^{-1}))^{-1/\alpha-1}$ is increasing for $u\in[0,ANe^{-1/\delta}]$ gives
\begin{align*}
    I_{1,(\ln N)/2}
    &\le(\ln(2AN/\ln N))^{-1/\alpha-1}\int_1^{(\ln N)/2} e^{-u}\,\dd u\\
    &\le(\ln(2AN/\ln N))^{-1/\alpha-1}\\
    &=o((\ln N)^{-1/\alpha})\quad\text{as }N\to\infty.
\end{align*}
Turning finally to the first integral, we have that
\begin{align*}
    I_{0,1}
    &=\int_0^1 u^{-1}(\ln(ANu^{-1}))^{-1/\alpha-1}\,\dd u\\
    &\quad+\int_0^1 u^{-1}(\ln(ANu^{-1}))^{-1/\alpha-1}(1-e^{-u})\,\dd u\\
    &=\alpha(\ln(AN))^{-1/\alpha}\\
    &\quad+\int_0^1 u^{-1}(\ln(ANu^{-1}))^{-1/\alpha-1}(1-e^{-u})\,\dd u.
\end{align*}
Furthermore, $1-e^{-u}\le u$ for $u\ge0$, and therefore
\begin{align*}
    0
    &\le\int_0^1 u^{-1}(\ln(ANu^{-1}))^{-1/\alpha-1}(1-e^{-u})\,\dd u\\
    &\le\int_0^1 (\ln(ANu^{-1}))^{-1/\alpha-1}\,\dd u\\
    &\le(\ln(AN))^{-1/\alpha-1}
    =o((\ln N)^{-1/\alpha})\quad\text{as }N\to\infty.
\end{align*}
We have thus proven that for any $A\ge1$,
\begin{align*}
    I_+(AN)\sim\alpha(\ln N)^{-1/\alpha}\quad\text{as }N\to\infty.
\end{align*}
Using the bounds in \eqref{Iplus} and \eqref{Iminus} and the fact that $\delta>0$ can be taken arbitrarily small completes the proof.
\end{proof}

We now use Lemma~\ref{lem:integralestimate} to prove Theorem~\ref{thm:uni}.

\begin{proof}[Proof of Theorem~\ref{thm:uni}]
    As $\E[Z]=\int_0^\infty\P(Z>z)\,\dd z$ for any nonnegative random variable $Z$, we thus have
    \begin{align*}
        \E[T_N^m]
        =\int_0^\infty\P(T_N>u^{1/m})\,\dd u
        =\int_0^\infty(S(u^{1/m}))^N\,\dd u,
    \end{align*}
    where $S(t):=\P(\tau>t)$ denotes the survival probability of $\tau$. 
    Let $\eps\in(0,1)$. By \eqref{eq:conditionb}, there exists a $\delta>0$ so that
    \begin{align*}
        1-e^{-C_-/t^\alpha}
        \le S(t)
        \le 1-e^{-C_+/t^\alpha}\quad\text{for all }t\in(0,\delta],
    \end{align*}
    where $C_\pm=(1\pm\eps)C$. Hence,
    \begin{align*}
        \int_0^\delta (1-e^{-C_-^m/u^{\alpha/m}})^N\,\dd u
        &\le\int_0^\delta (S(u^{1/m}))^N\,\dd u\\
        &\le\int_0^\delta (1-e^{-C_+^m/u^{\alpha/m}})^N\,\dd u.
    \end{align*}
    Lemma~\ref{lem:integralestimate} implies that as $N\to\infty$,
    \begin{align*}
        \int_0^\delta (1-e^{-C_\pm^m/u^{\alpha/m}})^N\,\dd u\sim C_\pm^{m/\alpha}(\ln N)^{-m/\alpha}.
    \end{align*}
    Since $\eps\in(0,1)$ is arbitrary, we thus have that
    \begin{align*}
        \int_0^\delta (S(u^{1/m}))^N\,\dd u\sim C^{m/\alpha}(\ln N)^{-m/\alpha}\quad\text{as }N\to\infty.
    \end{align*}

    By \eqref{eq:conditiona}, there exists an $N_{0}\ge1$ so that $\int_{0}^{\infty}(S(t))^{N_{0}}\,\dd t<\infty$. Thus, if $N_1=2^{m-1}N_{0}$, then
\begin{align*}
    \int_0^\infty(S(u^{1/m}))^{N_1}\,\dd u
    =m\int_{0}^{\infty}t^{m-1}(S(t))^{N_1}\,\dd t<\infty.
\end{align*}
Hence, if $N\ge N_{1}$, then since $S$ is nonincreasing,
\begin{align*}
\int_{\delta}^{\infty}(S(u^{1/m}))^{N}\,\dd u
\le K(S(\delta^{1/m}))^{N},
\end{align*}
and $S(\delta^{1/m})<1$ by \eqref{eq:conditionb}, and we have defined the finite constant,
\begin{align*}
K
=\int_{\delta}^{\infty}\Big(\frac{S(u^{1/m})}{S(\delta^{1/m})}\Big)^{N_{1}}\,\dd u<\infty.
\end{align*}
Therefore, $\int_\delta^\infty (S(u^{1/m}))^N\,\dd u$ vanishes exponentially as $N\to\infty$, which completes the proof.
\end{proof}

\section{\label{sec:shorttimeRLfbm}Bounded speed RLfBm FPT distribution at short-time}

Let $p_0$ and $p_1=1-p_0$ be the initial distribution of $J$,
\begin{align*}
    \P(J(0)=-1)
    =p_0,\quad \P(J(0)=1)=p_1=1-p_0.
\end{align*}
The position of $X=X_\eps$ in \eqref{eq:finitespeedprocess} at time $t\ge0$ is
\begin{align}\label{eq:X}
    X(t)
    =K\bigg[t^\beta+2\sum_{j=1}^{M(t)}(-1)^j \Big(t-\sum_{i=1}^j s_i\Big)^\beta\bigg],
\end{align}
where $\beta=(1+\alpha)/2$, $M(t)\in\{0,1,\dots\}$ is the number of jumps before time $t$, and $s_1$, $s_2$, \dots are the holding times of $J$ (and thus $J$ jumps at times $s_1$, $s_1+s_2$, $s_1+s_2+s_3$, \dots), and 
\begin{align*}
    K
    =(-1)^{(1-J(0))/2}K_0,\quad 
    K_0=\frac{\sqrt{\alpha}}{\beta}v \tc^{1-\beta}.
\end{align*}
In particular, $\{s_j\}_{j\ge1}$ are iid exponential random variables with rate $\lambda>0$.

We now determine the short-time distribution of the FPT $\tau=\tau_\eps$ in \eqref{eq:taueps}. Our approach modifies the approach in section~4.1 of Ref.~\cite{lawley2021pdmp}. Summing over the number of jumps before time $\tmin$ gives
\begin{align}\label{eq:sum}
\begin{split}
&\P(\tmin<\tau<\tmin(1+\eps))\\
&=\sum_{j=0}^{\infty}\P(\tmin<\tau<\tmin(1+\eps),{M}(\tmin)=j).
\end{split}
\end{align}
The $j=0$ term is zero for $\eps$ sufficiently small,
\begin{align}\label{eq:j0}
\P(\tmin<\tau<\tmin(1+\eps),{M}(\tmin)=0)=0,
\end{align}
because if ${M}(\tmin)=0$ and $J(0)=1$, then $\tau=\tmin$ and if ${M}(\tmin)=0$ and $J(0)=-1$, then $\tau\ge(1+\zeta)\tmin$ almost surely for some $\zeta>0$. Furthermore, it is straightforward to check that the terms with $j\ge3$ do not affect \eqref{eq:sum} to leading order as $\eps\to0$. 

We thus focus on the $j=1$ and $j=2$ terms in the sum in \eqref{eq:sum}. Looking first at the $j=1$ term, it is straightforward to check that if $\beta\le1$, then
\begin{align*}
\P(\tmin<\tau<\tmin(1+\eps),{M}(\tmin)=1,J(0)=1)=\O(\eps^{2}).
\end{align*}
It is also immediate that as $\eps\to0$,
\begin{align*}
&\P(\tmin<\tau<\tmin(1+\eps),{M}(\tmin)=1|J(0)=-1)\sim\\
&\P(\tmin<\tau<\tmin(1+\eps),{M}(\tmin(1+\eps))=1|J(0)=-1).
\end{align*}
If $\tmin<\tau<\tmin(1+\eps)$, ${M}(\tmin(1+\eps))=1$, and $J(0)=-1$, then \eqref{eq:X} implies
\begin{align*}
    &X(\tmin(1+\eps))\\
    &=-K_0\Big[\tmin^\beta(1+\eps)^\beta-2(\tmin(1+\eps)-s_1)^\beta\Big]>L,
\end{align*}
which is equivalent to
\begin{align*}
    s_1
    <\tmin\big(1+\eps-2^{-1/\beta}(1+(1+\eps)^\beta)^{1/\beta}\big)
    \sim\eps\tmin/2.
\end{align*}
Hence,
\begin{align*}
&\P(\tmin<\tau<\tmin(1+\eps),{M}(\tmin(1+\eps))=1|J(0)=-1)\\
&\quad\sim\P(s_{1}<\eps\tmin/2,s_{2}>\tmin(1+\eps)-s_{1}|J(0)=-1)\\
&\quad\sim\P(s_{1}<\eps\tmin/2,s_{2}>\tmin|J(0)=-1)\\
&\quad\sim\big[1-e^{-\lambda\eps\tmin/2}\big]e^{-\lambda\tmin}
\sim\lambda\tmin(\eps/2)e^{-\lambda\tmin}.
\end{align*}
Putting this together, we have that for $\beta\le1$,
\begin{align}\label{eq:j1}
\begin{split}
&\P(\tmin<\tau<\tmin(1+\eps),{M}(\tmin)=1)\\
&\sim p_{0}\lambda\tmin(\eps/2)e^{-\lambda\tmin}\quad\text{as }\eps\to0.
\end{split}
\end{align}

For $\beta>1$, we have that as $\eps\to0$,
\begin{align*}
    &\P(\tmin<\tau<\tmin(1+\eps),{M}(\tmin)=1|J(0)=1)\\
    &\sim\P(\tmin<\tau<\tmin(1+\eps),{M}(\tmin(1+\eps))=1|J(0)=1).
\end{align*}
Now, if $\tmin<\tau<\tmin(1+\eps)$, ${M}(\tmin(1+\eps))=1$, and $J(0)=1$, then \eqref{eq:X} implies
\begin{align*}
    &X(\tmin(1+\eps))\\
    &=K_0\Big[\tmin^\beta(1+\eps)^\beta-2(\tmin(1+\eps)-s_1)^\beta\Big]>L,
\end{align*}
which is equivalent to
\begin{align*}
    s_1/\tmin
    >1+\eps-2^{-1/\beta}((1+\eps)^\beta-1)^{1/\beta}.
\end{align*}
Hence, for $\beta>1$, 
\begin{align*}
    &\P(\tmin<\tau<\tmin(1+\eps),M(\tmin(1+\eps))=1|J(0)=1)\\
    &=\lambda\tmin e^{-\lambda\tmin(1+\eps)}\Big[2^{-1/\beta}\big((1+\eps)^\beta-1\big)^{1/\beta}\Big]_+\\
    &\sim\lambda\tmin e^{-\lambda\tmin}\Big(\frac{\beta}{2}\Big)^{1/\beta}\eps^{1/\beta}\quad\text{if }\beta>1.
\end{align*}

Moving to the $j=2$ term in \eqref{eq:sum}, we have
\begin{align*}
\P(\tmin<\tau<\tmin(1+\eps),{M}(\tmin)=2,J(0)=-1)
\end{align*}
does not contribute to leading order as $\eps\to0$, and
\begin{align*}
&\P(\tmin<\tau<\tmin(1+\eps),{M}(\tmin)=2|J(0)=1)\sim\\
&\P(\tmin<\tau<\tmin(1+\eps),{M}(\tmin(1+\eps))=2|J(0)=1).
\end{align*}
Now, if $\tmin<\tau<\tmin(1+\eps)$, ${M}(\tmin(1+\eps))=2$, and $J(0)=1$, then \eqref{eq:X} implies
\begin{align*}
    &X(\tmin(1+\eps))
    =K_0\Big[\tmin^\beta(1+\eps)^\beta-2(\tmin(1+\eps)-s_1)^\beta\\
    &\qquad\qquad\qquad\qquad\qquad+2(\tmin(1+\eps)-s_1-s_2)^\beta\Big]>L,
\end{align*}
which is equivalent to 
\begin{align*}
    s_2
    &<\tmin(1+\eps)-s_1\\
    &\quad-\Big[\big(\tmin(1+\eps)-s_1\big)^\beta-\frac{1}{2}(\tmin^\beta(1+\eps)^\beta-\tmin^\beta)\Big]^{1/\beta}\\
    &\sim(\eps/2)  \tmin^\beta (\tmin-s_1)^{1-\beta }\quad\text{as }\eps\to0.
\end{align*}
Hence, for $\beta\in(0,2)$, we have that
\begin{align*}
&\P(\tmin<\tau<\tmin(1+\eps),{M}(\tmin(1+\eps))=2|J(0)=1)\\
&\quad\sim(\lambda \tmin)^2 e^{-\lambda\tmin}\frac{1}{2(2-\beta)}\eps\quad\text{as }\eps\to0.
\end{align*}
Putting this together, we obtain that as $\eps\to0,$
\begin{align}\label{eq:j2}
\begin{split}
&\P(\tmin<\tau<\tmin(1+\eps),{M}(\tmin)=2)\\
&\quad\sim p_{1}(\lambda \tmin)^2 e^{-\lambda\tmin}\frac{1}{2(2-\beta)}\eps.
\end{split}
\end{align}

\subsection{Summary for $\beta\le1$}
Summarizing the $\beta\le1$ case, we have that
\begin{align*}
\P(\tmin<\tau<\tmin(1+\eps))
&\sim(1-q)A\eps\quad\text{as }\eps\to0,
\end{align*}
where
\begin{align*}
A
&=(1-q)^{-1}\Big(p_0+p_1\frac{\lambda \tmin}{2-\beta}\Big)\frac{\lambda\tmin}{2}e^{-\lambda\tmin}\nonumber\\
&=\Big(p_0+p_1\frac{\lambda \tmin}{2-\beta}\Big)\frac{\lambda\tmin}{2}\frac{e^{-\lambda\tmin}}{1-p_1e^{-\lambda\tmin}}.
\end{align*}

\subsection{Summary for $\beta>1$}

Summarizing the $\beta>1$ case, we have that
\begin{align*}
    \P(\tmin<\tau<\tmin(1+\eps))
    &\sim(1-q)A\eps^{1/\beta}\quad\text{as }\eps\to0,
\end{align*}
where
\begin{align*}
    A
    &=(1-q)^{-1}p_1\lambda \tmin e^{-\lambda \tmin}\Big(\frac{\beta}{2}\Big)^{1/\beta}\\
    &=\frac{p_1\lambda \tmin e^{-\lambda \tmin}}{1-p_1 e^{-\lambda\tmin}}\Big(\frac{\beta}{2}\Big)^{1/\beta}.
\end{align*}

\section{\label{sec:numericaldetails}Numerical simulation details}

We now describe the numerical methods used in section~\ref{sec:numerics}. We start with the sBm in section~\ref{sec:numericalsbm}.

\subsection{\label{sec:numercaldetailssbm}sBm}

Figures~\ref{fig:SBMcd}, \ref{fig:SBM}, and \ref{fig:SBMpdmp} use the exact FPT distribution formula for the unbounded speed sBm in \eqref{eq:moredetailed} and the following exact FPT distribution formula for the bounded speed sBm,
\begin{align*}
    \P(\tau_\eps>t)
    =S_\eps(\tc(t/\tc)^\alpha),
\end{align*}
where $S_\eps(t)$ is the following FPT distribution for the bounded speed sBm (and bounded speed RLfBm) when $\alpha=1$ \cite{malakar2018steady, grebenkov2026fastest},
\begin{align*}
    S_\eps(t)
    =1-\Theta(t-\tmin)\bigg[e^{-\kappa}+\kappa\int\limits_1^{t/\tmin}\frac{I_1(\kappa\sqrt{z^2-1})}{e^{\kappa z}\sqrt{z^2-1}}\,\dd z\bigg],
\end{align*}
where $\Theta(x)=1$ if $x\ge0$ and $\Theta(x)=0$ if $x<0$, $I_1$ is the modified Bessel function of the first kind, $\kappa=\lambda L/v$, and $\tmin=L/v$ is the minimal time in \eqref{eq:tmin} when $\alpha=1$.

\subsection{RLfBm}

The black curves and markers in Figure~\ref{fig:FBM} use the exact FPT distribution formulas described above in Appendix~\ref{sec:numercaldetailssbm}. Figure~\ref{fig:FBMpdmp} and the $\alpha\neq1$ plus markers in Figure~\ref{fig:FBM} use stochastic simulations of the bounded speed RLfBm in \eqref{eq:finitespeedprocess}. We now describe the stochastic simulation method.

For Figure~\ref{fig:FBM}, we use the representation of the process in \eqref{eq:X} and simulate the exponential times $s_1,s_2,\dots$ until either $X(\sum_{i=1}^n s_i)>L$ or the time exceeds $100\tmin$. If the former occurs, then we numerically solve for the time $\tau\in[\sum_{i=1}^{n-1} s_i,\sum_{i=1}^n s_i]$ such that $X(\tau)=L$. This method uses the fact that $X$ is a monotone function of time between jumps if $\alpha\in(0,1]$. We then repeat this method to obtain $M\gg1$ realizations of the FPT $\tau$, denoted $\{\tau_1,\dots,\tau_M\}$. Then, for a given value of $N\ge1$, we obtain $K=\floor*{M/N}$ iid realizations of the fFPT $T_{N}$ (where $\floor*{\cdot}$ denotes the floor operator) denoted by $T_{N}^{(1)},\dots,T_{N}^{(K)}$ by defining the following fFPT for $k\in\{1,\dots,K\}$,
\begin{align*}
T_{N}^{(k)}
:=\min\{\tau_{(k-1)N+1},\tau_{(k-1)N+2},\dots,\tau_{kN+1}\}.
\end{align*}
From these $K$ realizations of $T_{N}$, we then compute the mean fFPT in Figure~\ref{fig:FBM}.

For Figure~\ref{fig:FBMpdmp}, we similarly compute realizations of the fFPT from realizations of the FPT with two key differences. The first key difference is that we compute realizations of the FPT $\tau$ that are conditioned on being greater than the minimal time $\tmin$ in \eqref{eq:tminrlfbm}, which we denote by $\taut$. To obtain this conditioned FPT, we first simulate $J(0)$ conditioned that $\tau>\tmin$. If the unconditioned distribution of $J(0)$ is $\P(J(0)=-1)=p_0=1-p_1$, then the conditioned distribution of $J(0)$ is
\begin{align*}
    \P(J(0)=1|\tau>\tmin)
    &=\frac{\P(J(0)=1,\tau>\tmin)}{\P(\tau>\tmin)}\\
    &=\frac{p_1(1-e^{-\lambda\tmin})}{p_0+p_1(1-e^{-\lambda\tmin})},
\end{align*}
and $\P(J(0)=-1|\tau>\tmin)=1-\P(J(0)=1|\tau>\tmin)$. If $J(0)=-1$, then $s_1$ is exponentially distributed with rate $\lambda$. If $J(0)=1$, then $s_1$ is exponentially distributed with rate $\lambda$ conditioned that $s_1<\tmin$, which we simulate via
\begin{align*}
    s_{1}=-\ln(U(1-e^{-\lambda\tmin})+e^{-\lambda\tmin})/\lambda,
\end{align*}
where $U$ is uniformly distributed on $[0,1]$.

The next key difference is that for $\alpha>1$, the process $X$ is no longer guaranteed to be a monotone function of time in between jumps. We therefore simulate the holding times $s_1,s_2,s_3,\dots$ until the sum $\sum_i s_i$ exceeds $100\tmin$ and then we check if $X(s)>L$ by checking 20 evenly distributed values of $s$ between each pair of jump times. If this occurs, then we numerically solve for the time $\tau$ such that $X(\tau)=L$.

\subsubsection*{Acknowledgments}
SDL was supported by the National Science Foundation (Grant Number DMS-2325258). 


\bibliography{library.bib}

@article{hass2024extreme,
  title={Extreme diffusion measures statistical fluctuations of the environment},
  author={Hass, Jacob B and Drillick, Hindy and Corwin, Ivan and Corwin, Eric I},
  journal={Physical Review Letters},
  volume={133},
  number={26},
  pages={267102},
  year={2024},
  publisher={APS}
}

@article{hass2024first,
  title={First-passage time for many-particle diffusion in space-time random environments},
  author={Hass, Jacob B and Corwin, Ivan and Corwin, Eric I},
  journal={Physical Review E},
  volume={109},
  number={5},
  pages={054101},
  year={2024},
  publisher={APS}
}

@article{hass2023anomalous,
  title={Anomalous fluctuations of extremes in many-particle diffusion},
  author={Hass, Jacob B and Carroll-Godfrey, Aileen N and Corwin, Ivan and Corwin, Eric I},
  journal={Physical Review E},
  volume={107},
  number={2},
  pages={L022101},
  year={2023},
  publisher={APS}
}

@article{karamched2026entropic,
  title={Entropic Collapse and Extreme First-Passage Times in Discrete Ballistic Transport},
  author={Karamched, Bhargav R},
  journal={arXiv preprint arXiv:2601.03622},
  year={2026}
}

@article{carroll2025measurements,
  title={Measurements of extreme first passage times in photon transport},
  author={Carroll-Godfrey, Aileen N and Corwin, Eric I},
  journal={arXiv preprint arXiv:2502.19359},
  year={2025}
}

@article{kosztolowicz2022first,
  title={First-passage time for the g-subdiffusion process of vanishing particles},
  author={Koszto{\l}owicz, Tadeusz},
  journal={Physical Review E},
  volume={106},
  number={2},
  pages={L022104},
  year={2022},
  publisher={APS}
}

@article{rangarajan2000anomalous,
  title={Anomalous diffusion and the first passage time problem},
  author={Rangarajan, Govindan and Ding, Mingzhou},
  journal={Physical Review E},
  volume={62},
  number={1},
  pages={120},
  year={2000},
  publisher={APS}
}

@incollection{lawley2024competition,
  title={Competition of many searchers},
  author={Lawley, Sean D},
  booktitle={Target Search Problems},
  pages={281--303},
  year={2024},
  publisher={Springer}
}

@article{davis1984piecewise,
  title={Piecewise-deterministic Markov processes: A general class of non-diffusion stochastic models},
  author={Davis, Mark HA},
  journal={Journal of the Royal Statistical Society: Series B (Methodological)},
  volume={46},
  number={3},
  pages={353--376},
  year={1984},
  publisher={Wiley Online Library}
}

@inproceedings{hespanha2004stochastic,
  title={Stochastic hybrid systems: Application to communication networks},
  author={Hespanha, Joao P},
  booktitle={International Workshop on Hybrid Systems: Computation and Control},
  pages={387--401},
  year={2004},
  organization={Springer}
}

@article{dorsogna2026mean,
  title={Mean first passage times of velocity jump processes in higher dimensions},
  author={D'Orsogna, Maria R and Lindsay, Alan E and Hillen, Thomas},
  journal={arXiv preprint arXiv:2603.29241},
  year={2026}
}

@article{bena2006,
  title={Dichotomous {M}arkov noise: {E}xact results for out-of-equilibrium systems},
  author={Bena, I.},
  journal={Int. J. Mod. Phys. B},
  volume={20},
  number={20},
  pages={2825--2888},
  year={2006},
  publisher={World Scientific}
}

@article{malakar2018steady,
  title={Steady state, relaxation and first-passage properties of a run-and-tumble particle in one-dimension},
  author={Malakar, Kanaya and Jemseena, V and Kundu, Anupam and Vijay Kumar, K and Sabhapandit, Sanjib and Majumdar, Satya N and Redner, S and Dhar, Abhishek},
  journal={Journal of Statistical Mechanics: Theory and Experiment},
  volume={2018},
  number={4},
  pages={043215},
  year={2018},
  publisher={IOP Publishing and SISSA}
}

@article{lim2002self,
  title={Self-similar Gaussian processes for modeling anomalous diffusion},
  author={Lim, Soonchieh C and Muniandy, Sithi Vinayakam},
  journal={Physical Review E},
  volume={66},
  number={2},
  pages={021114},
  year={2002},
  publisher={APS}
}

@article{angelani2014first,
  title={First-passage time of run-and-tumble particles},
  author={Angelani, L and Di Leonardo, R and Paoluzzi, M},
  journal={The European Physical Journal E},
  volume={37},
  number={7},
  pages={59},
  year={2014},
  publisher={Springer}
}

@article{aurzada2022asymptotics,
	title = {Asymptotics of the {Persistence} {Exponent} of {Integrated} {Fractional} {Brownian} {Motion} and {Fractionally} {Integrated} {Brownian} {Motion}},
	volume = {67},
	issn = {0040-585X, 1095-7219},
	url = {https://epubs.siam.org/doi/10.1137/S0040585X97T990769},
	doi = {10.1137/S0040585X97T990769},
	language = {en},
	number = {1},
	urldate = {2026-03-28},
	journal = {Theory of Probability \& Its Applications},
	author = {Aurzada, F. and Kilian, M.},
	month = may,
	year = {2022},
	pages = {77--88},
}

@article{dean2021position,
  title={Position distribution in a generalized run-and-tumble process},
  author={Dean, David S and Majumdar, Satya N and Schawe, Hendrik},
  journal={Physical Review E},
  volume={103},
  number={1},
  pages={012130},
  year={2021},
  publisher={APS}
}

@incollection{lifshits2012lectures,
  title={Lectures on Gaussian processes},
  author={Lifshits, Mikhail},
  booktitle={Lectures on Gaussian Processes},
  pages={1--117},
  year={2012},
  publisher={Springer}
}

@article{mandelbrot1968fractional,
  title={Fractional Brownian motions, fractional noises and applications},
  author={Mandelbrot, Benoit B and Van Ness, John W},
  journal={SIAM review},
  volume={10},
  number={4},
  pages={422--437},
  year={1968},
  publisher={SIAM}
}

@article{joseph1989,
  title={Heat waves},
  author={Joseph, D D and Preziosi, L},
  journal={Rev Mod Phys},
  volume={61},
  number={1},
  pages={41},
  year={1989},
  publisher={APS}
}

@article{kuske1997,
  title={Large deviation theory for stochastic difference equations},
  author={Kuske, R and Keller, JB},
  journal={Eur J Appl Math},
  volume={8},
  number={6},
  pages={567--580},
  year={1997},
  publisher={Cambridge University Press}
}

@article{keller2004,
  title={Diffusion at finite speed and random walks},
  author={Keller, J B},
  journal={Proc Natl Acad Sci},
  volume={101},
  number={5},
  pages={1120},
  year={2004},
  publisher={National Academy of Sciences}
}

@article{grebenkov2026fastest,
  title={Fastest first-passage time for multiple searchers with finite speed},
  author={Grebenkov, Denis S and Metzler, Ralf and Oshanin, Gleb},
  journal={arXiv preprint arXiv:2602.15627},
  year={2026}
}

@book{grebenkov2024target,
  title={Target search problems},
  author={Grebenkov, Denis and Metzler, Ralf and Oshanin, Gleb},
  pages={1--29},
  year={2024},
  publisher={Springer}
}

@article{gomez2024first,
  title={First hitting time of a one-dimensional L{\'e}vy flight to small targets},
  author={Gomez, Daniel and Lawley, Sean D},
  journal={SIAM Journal on Applied Mathematics},
  volume={84},
  number={3},
  pages={1140--1162},
  year={2024},
  publisher={SIAM}
}

@article{lawley2023slow,
  title={Slowest first passage times, redundancy, and menopause timing},
  author={Lawley, Sean D and Johnson, Joshua},
  journal={Journal of Mathematical Biology},
  volume={86},
  number={6},
  pages={1--53},
  year={2023},
  publisher={Springer}
}

@article{lawley2020networks,
  title={Extreme first-passage times for random walks on networks},
  author={Lawley, Sean D},
  journal={Physical Review E},
  volume={102},
  number={6},
  pages={062118},
  year={2020},
  publisher={APS}
}

@article{lawley2023super,
  title={{Extreme statistics of superdiffusive L{\'e}vy flights and every other L{\'e}vy subordinate Brownian motion}},
  author={Lawley, Sean D},
  journal={Journal of Nonlinear Science},
  volume={33},
  number={4},
  pages={53},
  year={2023},
  publisher={Springer}
}

@article{lawley2021pdmp,
  title={Extreme first passage times of piecewise deterministic Markov processes},
  author={Lawley, Sean D},
  journal={Nonlinearity},
  volume={34},
  number={5},
  pages={2750},
  year={2021},
  publisher={IOP Publishing}
}

@article{lawley2020sub,
  title={Extreme statistics of anomalous subdiffusion following a fractional Fokker-Planck equation: Subdiffusion is faster than normal diffusion},
  author={Lawley, Sean D},
  journal={Journal of Physics A: Mathematical and Theoretical},
  year={2020},
  publisher={IOP Publishing},
  doi = {10.1088/1751-8121/aba39c},
}

@article{guigas2008,
  title={Sampling the cell with anomalous diffusion--the discovery of slowness},
  author={Guigas, Gernot and Weiss, Matthias},
  journal={Biophysical journal},
  volume={94},
  number={1},
  pages={90--94},
  year={2008},
  publisher={Elsevier}
}

@article{metzler2000,
  title={The random walk's guide to anomalous diffusion: a fractional dynamics approach},
  author={Metzler, Ralf and Klafter, Joseph},
  journal={Physics reports},
  volume={339},
  number={1},
  pages={1--77},
  year={2000},
  publisher={Elsevier}
}

@article{lawley2020dist,
  title={Distribution of extreme first passage times of diffusion},
  author={Lawley, S D},
  journal={Journal of Mathematical Biology},
  year={2020},
  doi={10.1007/s00285-020-01496-9},
  url = {http://dx.doi.org/10.1007/s00285-020-01496-9}
}

@article{condamin2008,
  title={Probing microscopic origins of confined subdiffusion by first-passage observables},
  author={Condamin, S and Tejedor, Vincent and Voituriez, Rapha{\"e}l and B{\'e}nichou, Olivier and Klafter, Joseph},
  journal={Proceedings of the National Academy of Sciences},
  volume={105},
  number={15},
  pages={5675--5680},
  year={2008},
  publisher={National Acad Sciences}
}

@article{condamin2007,
  title={First-passage time distributions for subdiffusion in confined geometry},
  author={Condamin, S and B{\'e}nichou, O and Klafter, J},
  journal={Physical review letters},
  volume={98},
  number={25},
  pages={250602},
  year={2007},
  publisher={APS}
}

@article{metzler1999,
  title={Anomalous diffusion and relaxation close to thermal equilibrium: A fractional {Fokker-Planck} equation approach},
  author={Metzler, Ralf and Barkai, Eli and Klafter, Joseph},
  journal={Physical review letters},
  volume={82},
  number={18},
  pages={3563},
  year={1999},
  publisher={APS}
}

@book{billingsley2013,
  title={Convergence of probability measures},
  author={Billingsley, P},
  year={2013},
  publisher={John Wiley \& Sons}
}

@article{corless1996,
  title={On the {LambertW} function},
  author={Corless, RM and Gonnet, GH and Hare, DEG and Jeffrey, DJ and Knuth, DE},
  journal={Advances in Computational mathematics},
  volume={5},
  number={1},
  pages={329--359},
  year={1996},
  publisher={Springer}
}

@article{lawley2020uni,
  title={Universal formula for extreme first passage statistics of diffusion},
  author={Lawley, S D},
  journal={Phys Rev E},
  volume={101},
  number={1},
  pages={012413},
  year={2020},
  publisher={APS}
}

@book{coles2001,
  title={An introduction to statistical modeling of extreme values},
  author={Coles, S and Bawa, J and Trenner, L and Dorazio, P},
  volume={208},
  year={2001},
  publisher={Springer}
}

@article{weiss1983,
  title={Order statistics for first passage times in diffusion processes},
  author={Weiss, G H and Shuler, K E and Lindenberg, K},
  journal={J Stat Phys},
  volume={31},
  number={2},
  pages={255--278},
  year={1983},
  publisher={Springer}
}

\end{document}